\documentclass[runningheads]{llncs}
\usepackage[T1]{fontenc}
\usepackage{graphicx}
\usepackage{amssymb}
\usepackage{amsmath}
\usepackage{xcolor}
\usepackage{psfrag}
\usepackage{subcaption}

\DeclareMathOperator*{\argmin}{arg\,min}

\newtheorem{assumption}{Assumption}

\begin{document}
\title{Synthesis of Limit Cycles and Reference Tracking via Switching Affine Systems}
%
%
\author{Nils Hanke\orcidID{0009-0008-8940-2677} \and
Zonglin Liu\orcidID{0000-0002-0196-9476} \and
Olaf Stursberg\orcidID{0000-0002-9600-457X}}
\authorrunning{N. Hanke et al.}
%
\institute{Control and System Theory, Dept. of Electrical Engineering and Computer Science, University of Kassel, Germany.\\
\email{\{n.hanke,z.liu,stursberg\}@uni-kassel.de}}
\maketitle              
\begin{abstract}
This paper introduces a novel method to approximate limit cycles of nonlinear ODEs by use of switching affine dynamics in order to ease data-based modeling and analysis. Previous approaches to approximating limit cycles by switching systems have been largely confined to simple partitions into two-regions or low-dimensional (often planar) settings. In contrast, this study utilizes more general partitions in higher-dimensional state spaces, augmented by external signals, to develop a synthesis scheme that guarantees a globally stable limit cycle. The synthesis task is formulated and solved based on constrained numerical optimization. Starting from sampled data of the nonlinear dynamics, the method minimizes the error between the data and the limit cycle generated by the switching affine model, while employing stability constraints to ensure global stability.
Based on the obtained model, the paper tackles the problem of reference tracking for switching affine systems with periodic behavior. While the approximation scheme is based on a common Lyapunov function, the reference tracking approach uses multiple Lyapunov functions to achieve less conservative convergence results. The principle and effectiveness of the proposed methods are illustrated through a set of examples.

\end{abstract}
\section{Introduction}
\label{section1}
Periodic behavior represents a foundational phenomenon in diverse fields ranging from biology over engineering to physics, among other disciplines \cite{teplinsky2008limit,mirollo1990synchronization,peterchev2003quantization}. Often periodic behavior is mathematically represented by established nonlinear oscillator models—including those by Kuramoto, Van der Pol, FitzHugh Nagumo, Duffing, and Goodwin \cite{kuramoto2005self,joshi2016synchronization,gaiko2011multiple,Dorfler.2014,kudryashov2021generalized,gonze2021goodwin,atherton1980survey} -- the scope of analytic techniques for these dynamics remains constrained. In particular, characterizing and analyzing limit cycles, especially with respect to uniqueness and stability conditions, is often feasible only in special cases. A key challenge therefore lies in transforming or approximating the underlying oscillatory dynamics by model class which are tractable for rigorous verification of dynamic properties. While existing data-driven techniques, including machine learning and hybrid system identification, are capable of approximating periodic dynamics \cite{xu2019frequency}, they typically lack the structure required for a formal analysis of limit cycle properties. This limitation impedes the systematic investigation of oscillatory phenomena in critical applications. One prominent example is the study of circadian rhythms in biological systems \cite{werckenthin2020neither}, where a thorough understanding of stability, phase shifts, and synchronization is indispensable.

Switching or piecewise-affine systems (PAS) are an effective model class for approximating nonlinear dynamics, facilitating simplified analysis \cite{paoletti2007identification,lauer2011continuous}. This effectiveness stems from two key features: (1) an analytical solution exists within each region of the partitioned state space, and (2) the approximation accuracy can be refined by adjusting the partition and the system parametrization. 
Regarding the approximation of nonlinear systems exhibiting periodic trajectories, \cite{lum1991global,freire1998bifurcation} established conditions for the existence of limit cycles in PAS defined by partitions consisting of two regions in the plane ($\mathbb{R}^2$). Subsequent studies \cite{coll2001degenerate,llibre2008existence} further investigated the uniqueness and stability of such limit cycles. Building on these foundational results, \cite{kai2012limit} synthesized planar PAS with stable polygonal limit cycles, while later work \cite{HS23,HLS24} developed algorithms to generate planar switching affine systems that approximate given limit cycles with guarantees on uniqueness and local stability. A significant limitation of these approaches is their reliance on just two affine dynamics separated by a single line, which restricts the achievable accuracy of approximation. For a comprehensive overview of conditions for the existence of limit cycle in planar piecewise-linear systems, see \cite{freire1998bifurcation} or Chapter 5.1 of \cite{bernardo2008piecewise}. Although the recent contribution \cite{HLS25} enhances approximation quality by employing multiple partitions, its scope remains restricted to planar systems ($\mathbb{R}^2$) and does not provide guarantees for the global stability of the synthesized limit cycle.

The present work extends the work in \cite{icinco25} and the methods referred to before to higher-dimensional spaces, and it proposes a new method to design tracking controllers for periodic PAS. By leveraging the contraction principle from \cite{pavlov2007convergence}, the proposed synthesis method generates piecewise-affine systems with globally stable limit cycles. A core contribution is a novel technique for state-space partitioning and region-specific affine dynamics synthesis, which ensures both global stability and adjustable approximation accuracy. Stabilization for this class of systems is typically achieved through state feedback control combined with Lyapunov stability theory, see \cite{johansson2003piecewise,habets2004control}. In the proposed synthesis approach for approximating periodic behavior with switching affine systems, stability of the limit cycles is guaranteed via a common Lyapunov function valid across all regions. As shown in \cite{johansson1997computation,zhu2019multiple} this requirement presents a known difficulty, as such functions are often hard to identify and tend to yield conservative results. For cases where a common Lyapunov function is unavailable or formal limit-cycle guarantees cannot be established, the design of a tracking controller is proposed. This controller utilizes multiple Lyapunov functions to achieve asymptotic tracking of the target periodic orbit, following principles from prior work \cite{van2008tracking,miljkovic2025} on non-periodic reference tracking.

The remainder of this paper is organized as follows. Section 2 provides definitions of limit cycles, stability, and reference tracking in the context of switching affine systems. In Section 3, first the partitioning of the state space in arbitrary dimensions is addressed. This is followed by the proposal of a synthesis strategy, designed for high approximation accuracy of the data set and formulated as an optimization problem with constraints which guarantee global stability of the resulting limit cycle. The identification procedure is demonstrated through examples including cases in 2 and 3 dimensions. Section 4 is dedicated to the design of a controller for reference tracking of periodic switching affine systems, accompanied by a numerical tracking example. Finally, Section 5 concludes the paper and outlines directions for future work.

\section{Problem Description}
\label{section2}
The aim of the main procedure proposed in this paper is to reconstruct the limit cycles of a broad class of oscillatory systems with the following property: The underlying nonlinear dynamics, defined in $\mathbb{R}^{n_x}$, produces a smooth and stable limit cycle that 1.) is located within an $(n_x-1)$-dimensional manifold, 2.) oscillates around a virtual center point, and 3.) exhibits neither strong twisting nor self-intersections.

Let an ordered set $F := \{\tilde{x}_1, \tilde{x}_2, \dots, \tilde{x}_{n_F}\}$ of state samples $\tilde{x}_i \in \mathbb{R}^{n_x}$ be collected along the limit cycle of the nonlinear dynamics. For simplicity, assume that the sampling time $\Delta t$ is constant along the cycle; however, the method presented later remains valid also for non-uniform sampling intervals.  
The sampling is considered dense in the sense that  $n_F \gg n_x$, or equivalently, the sampling time $\Delta t $ is significantly smaller than the period $T = n_F \cdot \Delta t$ of the limit cycle, and $\Delta t < 1$.

Given the set \(F\), the main objective of this contribution is to introduce a method for constructing a dynamic model approximating the sampled limit cycle while preserving its characteristic properties. To achieve this, the class of switching affine systems is employed: Let $x(t) \in \mathbb{R}^{n_x}$ be the state at time $t \in \mathbb{R}$, and let $u(t) \in \mathbb{R}^{n_u \times 1}$ represent an input signal that is multiplied by a matrix $B \in \mathbb{R}^{n_x \times n_u}$ -- the input will later be used to achieve convergence properties. The state space $\mathbb{R}^{n_x}$ is assumed to be partitioned into a finite number of polytopes $P_i \subseteq \mathbb{R}^{n_x}$ with $i \in\{1, \ldots, n_P\}$. These polytopes are parametrized by $C_i \in \mathbb{R}^{1 \times {n_x}}$, $d_i \in \mathbb{R}$ according to the following relation:
\begin{align}  \label{eq:partition} 
P_i:=\{x\in  \mathbb{R}^{n_x} | C_ix  \ge d_i, C_{i+1}x < d_{i+1}\},~\bigcup_{i =1}^{n_P} P_i =\mathbb{R}^{n_x}.
\end{align}
For pairs $(A_i, b_i)$ with $A_i \in \mathbb{R}^{n_x \times n_x}$ and $b_i \in \mathbb{R}^{n_x \times 1}$, the affine dynamics assigned to each $P_i$ is given by:
\begin{align}  \label{eq:pwadef} 
\dot{x}(t)=A_ix(t) +b_i + B u(t), ~\text{for}~ x(t) \in P_i.
\end{align}
Consider a set of switching times $T_k = \{t_0, t_1, \ldots\}$ with initial time $t_0 = 0$. A trajectory $\bar{x}_{[0,\infty[}$ of system \eqref{eq:pwadef} starting in $x(t_0) = x_0$ represents the state evolution across a sequence of phases $[t_k,t_{k+1}]$ between two consecutive switching times. For each phase with $t \in [t_k,t_{k+1}]$, the pair $(A_i,b_i)$ in \eqref{eq:pwadef} is activated with $i$ satisfying $x(t) \in P_i$. A limit cycle, as a specific trajectory of \eqref{eq:pwadef}, is defined as follow:
\begin{definition}
\label{Definition1}{Limit Cycle}\\
A trajectory $\bar{x}^{*}_{\lbrack0,\infty\lbrack}$ of \eqref{eq:pwadef} is called \emph{limit cycle}, if a finite period $T\in\mathbb{R}_{>0}$ exists
such that for any point $x(t)\in \bar{x}^{*}_{\lbrack0,\infty\lbrack}$, $t\in\mathbb{R}_{\geq 0}$ it applies that: $x(t+T)=x(t)$. \hfill $\Box$  
\end{definition}\vspace{1mm}

\begin{definition}
\label{Definition2}{Stability of a Limit Cycle}\\
A limit cycle $\bar{x}^{*}_{\lbrack0,\infty\lbrack}$ of \eqref{eq:pwadef} is called \emph{globally stable}, if every trajectory converges towards $\bar{x}^{*}_{\lbrack0,\infty\lbrack}$ independent of the initialization $x(0)=x_0\in \mathbb{R}^{2}$.
 \hfill $\Box$
\end{definition}
The identification of a model of the type  \eqref{eq:pwadef} based on the set $F$ requires the synthesis of the following parameters: 1.) the number $n_P$ of elements $P_i$ in the state space partition, 2.) the boundaries $C_i x = d_i$ for each $P_i$, 3.) the matrix pair $(A_i, b_i)$ corresponding to each $P_i$, and 4.) the signal $u(t)$ along with the vector $B$.
To satisfy the properties assumed for the limit cycle of the nonlinear system, the specific synthesis requirements are:
\begin{itemize}
\item The trajectory produced by system \eqref{eq:pwadef} must also constitute a globally stable limit cycle $\bar{x}^{*}_{[0,\infty[}$, as defined in Def. \ref{Definition2}.
\item The period of the limit cycle $\bar{x}^{*}_{[0,\infty[}$ must be $T = n_F \cdot \Delta t$.
\item The limit cycle $\bar{x}^{*}_{[0,\infty[}$ approximates the sample points in the set $F$ as closely as possible.
\end{itemize}
Switching affine systems represent a suitable candidate for such an approximation, as the number of parameters within each $P_i$ is small and the associated dynamics remain relatively simple to analyze, thus allowing to achieve the stability property. In particular, when the considered dynamics exhibits strong nonlinearity along the cycle, approximating it with affine dynamics within restricted regions is well justified, while freedom in selecting the $P_i$ (both in terms of number and geometry) allows for achieving arbitrarily accurate approximations in principle. Nevertheless, most existing work on approximating limit cycles is limited either to two regions in the plane or provides only local stability guarantees. To address these limitations, the following exposition employs the concept of \emph{contractivity} to derive a synthesis procedure that achieves the stated properties of $\bar{x}^{*}_{[0,\infty[}$ via numerical optimization.

A second goal involves the design of a controller to track the periodic behavior of switching affine systems. Consider the following structure of a system to be controlled:
\begin{align}  \label{eq:pwadeftracking} 
\dot{x}_c(t)=A_ix_c(t) +b_i + B u_c(t), ~\text{for}~ x_c(t) \in P_i.
\end{align}
Assume that a reference trajectory $x_r(t)$ to be tracked by control of \eqref{eq:pwadeftracking} is the solution of:
\begin{align}  \label{eq:reference} 
\dot{x}_r(t)=A_jx_r(t) +b_j +Bu_r(t)~\text{for}~ x_r(t) \in P_j,
\end{align}
for an initial state  $x_r(0) \in  \mathbb{R}^{n_x}$. In here, let $u_r(t)$ represent a piecewise continuous feed-forward control signal generating the reference, allowing that $j\in\{1, \ldots, n_P\}$ temporarily differs from $i$. Under these assumptions, the following reference tracking problem is addressed:
\begin{definition}[Reference Tracking Problem]
For any initial state  $x_c(0) \in   \mathbb{R}^{n_x}$ and any reference $x_r(t)$ generated by \eqref{eq:reference}, determine a control law for $u_c(t)$ in  \eqref{eq:pwadeftracking} for which:
 \begin{align} \label{eq:converge_blended}
 \lim_{t \to +\infty}||x_c(t) - x_r(t)|| =0 .
\end{align}
\hfill $\Box$
\end{definition}

\section{Synthesis of Globally Stable Limit Cycles via Switching Affine Systems}
\label{synthesis}
\subsection{Method to Partition the State Space}
\label{Partition}
To fully exploit the degrees of freedom provided by system \eqref{eq:pwadef} in assigning distinct affine dynamics to the regions while tracking the set $F$ of samples in $\mathbb{R}^{n_x}$, a method for partitioning the state space into $n_p$ polytopes is introduced first. In particular, the cases of $n_x =2$, $n_x =3$, and $n_x> 3$ are addressed separately in the sequel. The idea of this section follows the principle that $C_1 x = d_1$ and $C_{2} x = d_{2}$, with $C_1, C_{2} \in \mathbb{R}^{1 \times n_x}$ and $d_1, d_{2} \in \mathbb{R}$, intersect (at a point $x_s$ in the case $n_x=2$, in a line for $n_x=3$, a plane for $n_x=4$, and so on). Consequently, two bounding rays can be defined:
\begin{align}
&R_1:= \{ x \in  \mathbb{R}^{n_x}  |  C_1x=d_1, C_{2}x  \le d_{2} \},  \label{eq:ray1} \\
& R_{2}:=  \{ x \in  \mathbb{R}^{n_x}   |  C_{2}x=d_{2},C_1x  \ge d_1 \}  \label{eq:ray2},
\end{align}
which intersect in a point $x_s$ with $C_1 x_s =d_1=C_2 x_s=d_2$. These rays form the boundaries of the polyhedral cone:
\begin{align}
P_1:=\{x \in  \mathbb{R}^{n_x} ~ | ~ C_1x  \ge d_1,~~  C_{2}x  < d_{2}\}.
\end{align} Based on this idea, the proposed procedure determines a specific state space partition $P_i\subseteq\mathbb{R}^{n_x},~i\in \{1,\dots,n_p\}$ in the form of \eqref{eq:partition}, which is  parametrized by $C_i\in \mathbb{R}^{1\times n_x}$, $d_i\in \mathbb{R}$ with  $C_{n_p+1}=C_1$ and $d_{n_p+1}=d_1$. This partition is particular in the sense that the finite number $n_p$ of polytopes $P_i$ and bounding planes $C_i,d_i$ is the same.\\ 

\underline{Case $n_x =2$}: Let a \textit{center point} $x_s$ of all points in $F$ be determined by:
\begin{align}\label{eq:middlepoint}
x_{s,[q]}=\frac{1}{2}\left(\max_{l \in \{1,\ldots, n_F\}} \tilde{x}_{l, [q]} - \min_{l \in \{1,\ldots, n_F\}} \tilde{x}_{l, [q]}\right)
\end{align}
for the two dimensions $q\in\{1, 2\}$. Assuming that $x_s$ does not coincide with any point in $F$, a subset of $n_p$ sample points $\hat{x}_{1},\ldots,  \hat{x}_{n_p}$ is selected from $F$.
While details on the selection of this subset can be found in Remark \ref{Remark2} at the end of this subsection, note at this point that for $n_F\gg n_x$ and a sufficiently small $\Delta t$, this subset can be selected to construct a polyhedral partition: For each point $\hat{x}_{i}$, a unique line $C_i x=d_i$ with $C_i \in \mathbb{R}^{1 \times 2}$, $d_i \in \mathbb{R}$ can be determined, which passes through both $x_s$ and $\hat{x}_{i}$ and serves as the boundary between two adjacent cones $P_{i-1}$ and $P_{i}$, as illustrated in Fig.~\ref{fig:partition}. The set of the $\hat{x}_{i}$ also partitions the sample set according to $F=\bigcup_{i\in{1,\ldots,n_p}} F_i$ where each ordered subset $F_i$ contains the samples from $\hat{x}_{i}$ to the one sample in $F$ before $\hat{x}_{i+1}$ (using $\hat{x}_1=\hat{x}_{n_p+1}$).\\

\begin{figure}[!t]
  \begin{center}
		\psfrag{s}[rc][rc][1]{$x_s$}
			\psfrag{x1}[rc][rc][1]{$\hat{x}_{1}$}	
   \psfrag{x2}[rc][rc][1]{$\hat{x}_{2}$}			
	   \psfrag{x3}[l][l][1]{$\hat{x}_{3}$}			
	      \psfrag{x4}[rc][rc][1]{$\hat{x}_{4}$}			
	         \psfrag{x5}[rc][rc][1]{$\hat{x}_{5}$}		
	         
				\psfrag{p1}[rc][rc][1]{$P_{1}$}	         
	 	\psfrag{p2}[rc][rc][1]{$P_{2}$}	
	 		\psfrag{p3}[rc][rc][1]{$P_{3}$}	
	 			\psfrag{p4}[rc][rc][1]{$P_{4}$}	
	 				\psfrag{p5}[rc][rc][1]{$P_{5}$}

 \includegraphics[width = 0.4\textwidth]{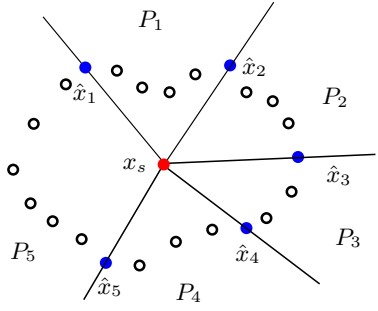}
    \caption{\cite{icinco25} Based on the set of samples $F$, the center point $x_s$, and a set of selected points $\hat{x}_{1}, \ldots, \hat{x}_{n_P} \in F$ (each representing the first state of any subset $F_1, \ldots, F_{n_P}$ along the limit cycle), the lines for partitioning the state space into regions $P_i$ are determined.}
    \label{fig:partition}
  \end{center}
\end{figure}

\underline{Case $n_x =3$}: Let the  center point $x_s$  again be determined by \eqref{eq:middlepoint}. Then,  a plane $\Omega^* x= \epsilon^*$ with normal vector $\Omega^*  \in  \mathbb{R}^{1 \times 3}$,  $\epsilon^* \in  \mathbb{R}$, is determined  by:
    \begin{align} \label{eq:optimalplane} 
   & (\Omega^*,  \epsilon^*):= \argmin_{\Omega,  \epsilon} \sum_{i=1}^{n_F} \|\Omega \tilde x_i- \epsilon \|_2,~~ \text{s.t.}~\Omega x_s=\epsilon
    \end{align}
Based on the outcome of \eqref{eq:optimalplane}, a line:
\begin{align}  \label{eq:3dline} 
\Gamma:=\{x \in  \mathbb{R}^3 ~| ~ x= x_s + \eta \Omega^*, ~ \eta \in \mathbb{R} \}
\end{align} 
which contains $x_s$ and has the orientation of  $\Omega^*$ is obtained. The plane $\Omega^* x= \epsilon^*$ is characterized by the fact that it contains $x_s$ and the accumulated distance between the sample points in $F$ to the plane is minimal. If no point in $F$ lies\footnote{If this condition does not hold, one may resolve the issue by slightly changing $\Omega^*$.} within $\Gamma$, a set of $n_p$ sample points $\hat{x}_{1},\ldots, \hat{x}_{n_p}$ is selected from $F$. For each of these points $\hat{x}_{i}$, a unique plane is determined that contains both $\hat{x}_{i}$ and the line $\Gamma$. Provided the plane does not contain any other sample point from $F$, it is  used to define the boundary between two adjacent regions $P_{i-1}$ and $P_{i}$. By this construction, a partition as the one shown in Fig.~\ref{fig:meinBild}b)
is obtained, while those are avoided which do not show a common intersection line, as in the example of Fig.~\ref{fig:meinBild}a).\\ 

 \begin{figure}[t!]
    \centering
    \psfrag{A}[][]{$a)$}
    \psfrag{B}[][]{$b)$}
    \psfrag{a}[][]{$x_s$}
    \psfrag{aa}[][]{$\Omega^* $}
    \psfrag{a1}[][]{$\Gamma$}
    \psfrag{a2}[][]{$P_{5}$}
    \psfrag{a3}[][]{$P_{6}$}
    \psfrag{a4}[][l]{$P_{1}$}
    \psfrag{a5}[][]{$P_{2}$}
    \psfrag{a6}[][]{$P_{3}$}
    \psfrag{a7}[][]{$P_{4}$}
    \psfrag{p1}[][l]{$\hat{x}_{1}$}
    \psfrag{p2}[][]{$\hat{x}_{2}$}
    \psfrag{p3}[][l]{$\hat{x}_{3}$}
    \psfrag{p4}[][]{$\hat{x}_{4}$}
    \psfrag{p5}[][l]{$\hat{x}_{5}$}
    \psfrag{p6}[][]{$\hat{x}_{6}$}
    \includegraphics[width=7.25cm,height=4.00cm]{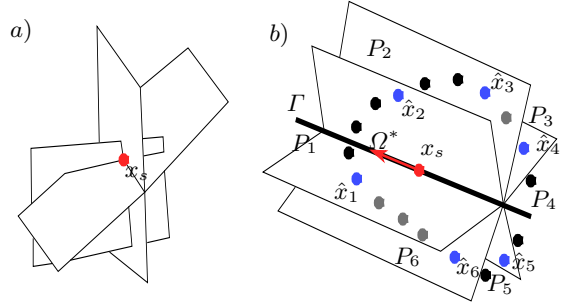} 
    \caption{\cite{icinco25} In the case $n_x=3$, the partition shown in Fig.~\ref{fig:meinBild} a) fails to meet requirement \eqref{eq:3dline} and the numbers of regions and partitioning planes are different. Conversely, the procedure described above yields an admissible partition for the scenario shown in Fig.~\ref{fig:meinBild} b).} 
    \label{fig:meinBild}
\end{figure}

\underline{Case $n_x >3$:} Determine again $x_s$, the vector $\Omega^*  \in  \mathbb{R}^{1 \times n_x}$, and the line $\Gamma$ according to \eqref{eq:middlepoint} to  \eqref{eq:3dline}. Also, let $n_p$ sample points $\hat{x}_{1},\ldots,  \hat{x}_{n_p}$ be selected from $F$. However, for any given sample point $\hat{x}_{i}$, the choice of a hyperplane defined as $n_x-2$-dimensional subspace containing $\hat{x}_{i}$ and the line $\Gamma$ is not unique.

To resolve this issue, a set of linearly independent  vectors $\Omega_2, \ldots, \Omega_{n_x-2}$ with $\Omega_j \in \mathbb{R}^{1 \times n_x}$, are identified, which have to be linearly independent of $\Omega^*$. Based on these vectors, a hyperplane $\Psi$ can be determined within the $n_x-2$ dimensional subspace as follows:
\begin{align}  \label{eq:highdimensionplane} 
\Psi:=\{x \in  \mathbb{R}^{n_x} |  x= x_s +  \eta_1 \Omega^*+ \sum\limits_{j=2}^{n_x-2}\eta_j\Omega_j, ~ \eta_j \in \mathbb{R} \}.
\end{align}
Assume that $\Psi$ does not contain any points from the set $F$. Then, for each sample point $\hat{x}_{i}$, a unique hyperplane within the $(n_x-1)$-dimensional subspace can be determined that contains both $\Psi$ and $\hat{x}_{i}$. Provided such a hyperplane does not contain any other sample point from $F$, it is then selected to be the boundary $R_{i-1,i}:= \{ x \in \mathbb{R}^{n_x} \mid C_i x = d_i\}$ between the adjacent regions $P_{i-1}$ and $P_{i}$.

\begin{remark}
\label{Remark1}
The proposed procedure excludes a-priori sets $F$ that are likely to result in poor approximations. One such case is illustrated in Fig.~\ref{fig:meinBild2} a), where the condition that $C_i x = d_i$ contains only one sample point $\hat x_i$ and the line $\Gamma$ is not satisfied. The second case, shown in Fig.~\ref{fig:meinBild2} b), corresponds to a limit cycle that intersects itself. Although the proposed partitioning procedure may succeed in this scenario, it becomes clearly evident (e.g., in region $P_3$) that no affine system can be identified that adequately captures the opposing directions of motion of the trajectories. This limitation motivates the exclusion of limit cycles with self-intersections, as stated in Section \ref{section2}. \hfill$\Box$
\end{remark}
\begin{figure}[!t]
    \centering
    \psfrag{x1}[][]{$x_1$}
    \psfrag{x2}[][]{$x_3$}
    \psfrag{x3}[][]{$x_2$}
    \psfrag{A}[][]{$a)$}
    \psfrag{B}[][]{$b)$}
    \psfrag{a}[][]{$x_s$}
    \psfrag{aa}[][]{$\Omega^* $}
    \psfrag{a1}[][]{$\Gamma$}
    \psfrag{a4}[][l]{$P_{1}$}
    \psfrag{a5}[][]{$P_{2}$}
    \psfrag{a6}[][]{$P_{3}$}
    \psfrag{a7}[][]{$P_{4}$}
    \psfrag{h1}[][]{$\hat{x}_{1}$}
    \psfrag{h2}[][]{$\hat{x}_{2}$}
    \psfrag{h3}[][]{$\Gamma$}
    \psfrag{h4}[][r]{$\Omega^* x= \epsilon^*$}
    \psfrag{h5}[][r]{$C_ix=d_i$}
    \includegraphics[width=7.25cm,height=4.00cm]{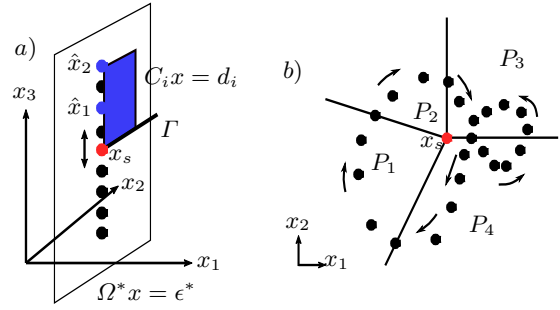} 
    \caption{
    \cite{icinco25} Certain configurations of $F$ are incompatible with the partitioning procedure. In example $a)$, the condition is violated that $C_ix=d_i$ should contain only a sample point $\hat x_i$ and $\Gamma$. In example $b)$, the identification of $A_3, b_3 \in P_3$ and $A_4, b_4 \in P_4$ is not feasible.} 
    \label{fig:meinBild2}
\end{figure}

\begin{remark}
\label{Remark2}
 The $n_p$ sample points of the sequence $\{\hat{x}_i\}_{i=1}^{n_p}$ with timestamps $\{\hat t_i\}_{i=1}^{n_p}$ are selected from $F$ such that all time intervals satisfy:
\begin{equation}
\Delta \hat t_i := |\hat t_{i+1} - \hat t_i| < 1 \quad \forall i \in \{1,\dots,n_p\}
\end{equation} with $\hat t_{n_p+1}\equiv t_1$.
Although this choice is not strictly required for partitioning the state space, it enables an approximation of the reconstructed limit cycle with relatively high quality, as will be further detailed in in Sec.~\ref{approximation}. \hfill$\Box$
\end{remark}

\subsection{Construction of the Dynamics}
\label{contractive}
Let the number $n_F$ and the partition of the state space into regions $P_i$ be obtained from the procedure in section \ref{Partition}. Following \cite{pavlov2007convergence}, the switching affine system \eqref{eq:pwadef} is termed \emph{contractive} if the following conditions are satisfied:
\begin{itemize}
\item Condition 1: $A_ix+b_i = A_{i+1}x+b_{i+1}$ holds for all $x$ on the boundary $C_ix  = d_i$ for $i \in\{1, \ldots, n_P\}$ and with $A_{n_p+1}=A_1$,  $b_{n_p+1}=b_1$.
\item Condition 2:  $A^T_iQ+ QA_i \prec 0$, $i \in\{1, \ldots, n_P\}$ holds for a positive-definite matrix $Q \succ 0$.
\end{itemize} 
The first condition requires the continuity of the gradient of the autonomous dynamics on the switching boundaries, whereas the second condition implies the existence of a common Lyapunov function valid in all regions. The property of contractivity, on which the synthesis procedure in this paper is based, is now set in relation to limit cycles by the following result from literature:
\begin{lemma}  \label{lemma:contractive} (Contractive switching affine systems \cite{demidovich1967lectures,pavlov2007convergence})
If the system \eqref{eq:pwadef}  is contractive with a non-zero vector  $B$, then for any piecewise continuous periodic signal $u(t)$ with a period $T$, the solution $x(t)$ with $t\geq 0$ starting from an arbitrary $x(0)\in \mathbb{R}^{n_x}$ always converges to a unique limit cycle with the same period  $T$.  
\hfill$\Box$ 
\end{lemma}
To encode the requirement of continuous gradients on the switching boundaries, the following equality constraints for the synthesis of $(A_i,b_i)$, $i \in\{1, \ldots, n_P\}$, are proposed:
\begin{align}\label{eq:continuousgradient}
A_{i,[q]}\hspace{-0.8mm}- \hspace{-0.8mm}A_{i+1,[q]} \hspace{-0.8mm}= \hspace{-0.8mm}\alpha_{i,[q]} C_i,~~b_{i,[q]}\hspace{-0.8mm}- \hspace{-0.8mm}b_{i+1,[q]} \hspace{-0.8mm}= \hspace{-0.8mm}\alpha_{i, [q]}d_i
\end{align}
for  $\alpha_{i, [q]} \in  \mathbb{R}$, $q \in\{1,\ldots, n_x\}$, and with $A_{i,[q]}$ representing the $q$-th row of $A_i$. Here, $\alpha_{i, [q]}$ is a scalar degree of freedom, used in the subsequent optimization, which allows the dynamics to differ on either side of the switching boundaries while guaranteeing a continuous transition.
The condition for the existence of a common Lyapunov function constitutes a nonlinear matrix inequality involving the matrices $A_1, \ldots, A_{n_P}$ and $Q$. If these constraints are satisfied, system \eqref{eq:pwadef} is guaranteed to possess a globally stable limit cycle as defined in Def.~\ref{Definition1} and Def.~\ref{Definition2} with a period $T$, provided the signal $u(t)$ is also periodic with the same period length (see Lemma \ref{lemma:contractive}). Thus, the task to be solved is one of determining the system \eqref{eq:pwadef} such that its limit cycle approximates the sample points in $F$ in terms of position and time.

\subsection{Approximation of the Sample Set $F$}
\label{approximation}
To explain the approximation of the subset $F_i$ by the dynamics assigned to $P_i$, the procedure is exemplarily described for $F_1$. The scheme equivalently transfers to the remaining sets $F_i$, $i\in\{2,\ldots,n_P\}$. For the set $F_1=\{\hat{x}_{1}, \tilde{x}_{2}, \ldots, \tilde{x}_{n_1}\}$ containing $n_1$ points and for $\hat{x}_{1} = \tilde{x}_{1}$, the following cost functional is defined:
\begin{align}  \label{eq:localcost} 
&J_1:=\sum_{j=2}^{n_1} ||   e^{A_1 (j-1)\Delta t}\hat{x}_1 +\hspace{-0.8mm}\int_{0}^{ (j-1)\Delta t}\hspace{-0.8mm}\hspace{-0.8mm}e^{A_1( (j-1)\Delta t- \tau)} (b_1 +B u(\tau))  ~d \tau    -            \tilde{x}_{j}      ||^2_2.  
\end{align}
It records the  difference between the reachable points of  $\dot{x}(t)=A_1x(t) +b_1+Bu(t)$ (starting from $\hat{x}_{1}$) and the sampled points in  $F_1$ at each  sampling time\footnote{For sampled states in $F$ with non-uniform but known sampling times, only the corresponding times in \eqref{eq:localcost} need to be adjusted.}.
The costs $J_1$ are minimized for a given signal $u(t)$ in order to synthesize $A_1$, $b_1$, and $B$. The solution is challenging due to the nonlinearity introduced by the matrix exponential function $e^{A_1 t}$. According to the Taylor series:
\begin{align}  \label{eq:taylor} 
e^{A_1t} \hspace{-0.8mm}= \hspace{-0.8mm} I_{n_x}  \hspace{-0.8mm}+  \hspace{-0.8mm}A_1t \hspace{-0.8mm} +  \hspace{-0.8mm}\frac{1}{2!}A^2_1t^2  \hspace{-0.8mm}+  \hspace{-0.8mm}\ldots  \hspace{-0.8mm}\approx  \hspace{-0.8mm}I_{n_x}  \hspace{-0.8mm}+  \hspace{-0.8mm} \sum_{j=1}^{n_d} \frac{1}{j!}A^j_1t^j,
\end{align}
the value of $e^{A_1t}$ can be approximated by the right-hand side of 
\eqref{eq:taylor} with  sufficiently high order $n_d$. For higher state dimensions $n_x$, a large value $n_d$ would, however, significantly increase the complexity of the optimization due to a higher-order nonlinearity, while a smaller order (such as $n_d \le 3$) would possibly result in considerable approximation errors. To address this issue, countermeasures based on the following observations are considered:
\begin{itemize}
\item For a fixed matrix $A_1$, the approximation error in \eqref{eq:localcost} remains small for small times $t$. In particular for $t < 1$, the series $t^j$, $j \in \{n_d, n_d+1, \ldots\}$ in the neglected terms converges to zero, meaning that omitting these terms contributes only very little to the approximation error.
 
\item For a fixed time $t$ and if the spectrum of $A_1$ lies within the unit circle, the matrices $A^j_1$, $j \in \{n_d, n_d+1, \ldots\}$, in the neglected terms also converge to zero, what as well results in small errors.
\end{itemize}

The first countermeasure is included into the partitioning procedure described in Sec.~\ref{Partition}, Remark \ref{Remark2} by selecting sample points on the boundaries such that the condition $j \cdot \Delta t < 1$ holds for all $j \in\{ 1, \ldots, n_1\}$ in equation \eqref{eq:localcost}. Note that a larger number $n_p$ of regions reduces in average the transition time from one boundary to the next. This leads to smaller approximation errors in the Taylor series expansion for a given order $n_d$, albeit at the cost of requiring the synthesis of more pairs $(A_i, b_i)$. The second countermeasure is established by ensuring that the largest singular value of $A_1$ is less than one, a condition that can be guaranteed by enforcing the nonlinear matrix inequality:
\begin{align}  \label{eq:singularvalue} 
A^T_1A_1 \prec I_{n_x}.
\end{align}
The condition \eqref{eq:singularvalue} forces the eigenvalues of $A_1$ to lie inside the left half of the unit circle, as the contraction condition additionally requires the real parts of the eigenvalues of $A_1$ to be negative. As a result, the convergence rate of system \eqref{eq:pwadef} in region $P_1$ is also bounded by one. This can be counterproductive if the sample points to be tracked in $F_1$ indicate that the state of the sampled limit cycle changes at a significantly different rate in a particular region of the state space. Therefore, the inclusion of \eqref{eq:singularvalue} should be considered an optional measure, or it should be replaced by a less conservative condition, such as $A^T_1 A_1 \prec \beta I_{n_x}$ for some $\beta > 1$. The described countermeasures affect only the approximation quality without compromising the contraction property guaranteed by Lemma \ref{lemma:contractive}. A detailed analysis of the upper bound on the approximation error can be found in \cite{higham2009scaling,kenney1998schur}.

The periodic signal $u(t)$ can be chosen, e.g., as a  piecewise constant scalar ($n_u=1$) for $k\in\{0,1,2,\ldots\}$:
\begin{align}  \label{eq:signal} 
u(t)=\begin{cases}
-1,~t<[kT, (k+\frac{1}{2})T)\\
1,~t<[(k+\frac{1}{2})T, (k+1)T).
    \end{cases}
\end{align}
The integral part of \eqref{eq:localcost} then leads to  the analytic expression:
\begin{align*}  
&\int_{0}^{t}e^{A_1(t-\tau)} (b_1 +B u(\tau))  ~d \tau = \notag\\
&\begin{cases}
 (e^{A_1t}-I_{n_x})A_1^{-1}(b_1 \hspace{-0.8mm}- \hspace{-0.8mm}B) ~\text{for} ~ 0 \le t <\frac{1}{2}T \\
 (e^{A_1\frac{T}{2}} \hspace{-0.8mm}- \hspace{-0.8mm}I_{n_x})A_1^{-1}(b_1 \hspace{-0.8mm}- \hspace{-0.8mm}B) \hspace{-0.8mm} + \hspace{-0.8mm}  (e^{A_1(t-\frac{1}{2}T)} \hspace{-0.8mm}- \hspace{-0.8mm}I_{n_x}) A_1^{-1}(b_1 \hspace{-0.8mm}+ \hspace{-0.8mm} B)
  ~\text{for} ~\frac{1}{2}T \le t <T,
    \end{cases}
\end{align*}
what is particularly beneficial for applying the previously discussed countermeasures in synthesis, see Section~\ref{optimization}. 
Note that the matrix $A_1$ is always invertible, as  the contractivity condition  implies that $A_1$ must be Hurwitz.

\subsection{Overall Optimization Problem}
\label{optimization}
Assuming the state space is partitioned as described in Section \ref{Partition} and that $u(t)$ is a periodic signal according to \eqref{eq:signal}, the synthesis of the switching affine system \eqref{eq:pwadef} can be carried out by solving the following optimization problem:
\begin{align}
    &~~~~~~ \min_{A_i, b_i, \alpha_i, i \in\{1,\ldots,n_p\}, B, Q} \sum_{i=1}^{n_P} J_i \label{OptP1}\\
   & \text{s.t.}~ \text{for all}~  i \in\{1,\ldots,n_p\}: \notag \\
 & \text{constraints} ~\eqref{eq:continuousgradient} ~\forall q \in \{1,\ldots, n_x\}, \label{OptP1c1}\\
     & A^T_iQ+ QA_i \prec 0,~Q \succ 0, ~B \ne 0, \label{OptP1c2}\\
     &A^T_iA_i \prec I_{n_x}~\text{(optional constraint)},\label{OptP1c3}\\
    & e^{A_i n_i \Delta t}\hat{x}_i \hspace{-0.8mm} +  \hspace{-1.2mm}\int_{0}^{  n_i\Delta t} \hspace{-1.2mm}e^{A_i( n_i \Delta t - \tau)} (b_i  \hspace{-0.8mm}+  \hspace{-0.8mm}B u(\tau))  d \tau   \hspace{-0.8mm} =  \hspace{-0.8mm}\hat{x}_{i+1}.  \label{OptP1c4}
\end{align}
In this context, $J_i$ is defined for each region in  in the same form as $J_1$ in \eqref{eq:localcost}. This constitutes a nonlinear optimization problem with a total of $(n_P+1)n^2_x + (2n_P+1)n_x$ variables. The cost functional minimized in problem \eqref{OptP1} measures the distance between each sample point in $F$ and the corresponding point on the limit cycle of system \eqref{eq:pwadef} at the same time instant. The constraints \eqref{OptP1c1} and \eqref{OptP1c2} guarantee that the resulting system \eqref{eq:pwadef} is contractive. The optional constraint \eqref{OptP1c3} aims to achieve a satisfactory approximation by employing the Taylor expansion for the matrix exponential function in expressions \eqref{OptP1} and \eqref{OptP1c4}. The final constraint aims to ensure that the set of sample points $\hat{x}_{i}$, $i \in\{1,\ldots,n_P\}$ located on the boundaries are reached by the limit cycle of system \eqref{eq:pwadef}. Otherwise, the minimization of objective \eqref{OptP1} may only lead to a transient trajectory of system \eqref{eq:pwadef} to follow the sample points in $F$, rather than the limit cycle of \eqref{eq:pwadef}, due to the approximation error inherent in $e^{A_i t}$. In general, the optimization problem is not guaranteed to be feasible; however, increasing the number $n_p$ or selecting a different partition can contribute to finding a feasible solution. As an extension to account for transient behavior from arbitrary initial states converging to the limit cycle, an additional term can be included into \eqref{OptP1} to minimize the deviation of the model dynamics from the sampled transient trajectories. Finally, it is worth noting that measurement noise, which may have perturbed the samples in $F$, is effectively eliminated by solving the optimization problem defined by \eqref{OptP1} to \eqref{OptP1c4}, and by assigning the model \eqref{eq:pwadef} to any of the regions $P_i$.

\subsection{Numeric Examples in 2D and 3D}
The effectiveness of the proposed synthesis approach is assessed by first considering an example with a set $F$ consisting of nineteen points in the plane, shown as black circles in Fig.~\ref{Example1} and Fig.~\ref{Example2}. The period extracted from the sample set is $T=1.4$, which is also adopted for the periodic signal defined in \eqref{eq:signal}. Following the rules of Section~\ref{Partition}, the state space is initially divided by eight rays (depicted as solid blue lines), with:
\begin{align*}
   & C_1=\begin{bmatrix}
    1&-\frac{1}{4}
\end{bmatrix},~C_2=\begin{bmatrix}
    1&-2
\end{bmatrix}, ~
C_3=\begin{bmatrix}
    -1&-2
\end{bmatrix}, ~
C_4=\begin{bmatrix}
    -1&-\frac{1}{4}
\end{bmatrix},\\
&C_5=\begin{bmatrix}
    -1&\frac{1}{5}
\end{bmatrix},~
C_6=\begin{bmatrix}
    -1&3
\end{bmatrix}, ~
C_7=\begin{bmatrix}
    1&\frac{5}{2}
\end{bmatrix},~
C_8=\begin{bmatrix}
    1&\frac{1}{6}
\end{bmatrix}, \\
&d_1=\frac{1}{2},~d_2=-3, ~d_3=\hspace{-1mm}-5,~d_4=\hspace{-1mm}-\frac{3}{2}, ~d_5=\hspace{-1mm}-\frac{3}{5},~d_6=\hspace{-1mm}5,~d_7=\hspace{-1mm}6,~d_8=\hspace{-1mm}\frac{4}{3}
\end{align*}
using the center point $x_s=[1,2]^T$. Figure~\ref{Example1} shows the limit cycle of system \eqref{eq:pwadef} resulting from the solution of \eqref{OptP1}, with an approximation order of $n_d=4$ used for the matrix exponential function. The shape of the limit cycle provides clear evidence that a small $n_d$ introduces substantial approximation errors.
\begin{figure}[!b]
    \centering
    \small
    \psfrag{a}[][]{$x_1$}
    \psfrag{a1}[][]{$x_2$}
    \psfrag{50}[][]{$\scriptstyle 50$}
    \psfrag{-50}[][]{$\scriptstyle -50$}
    \psfrag{0}[][]{$\scriptstyle 0$}
    \psfrag{5}[][]{$\scriptstyle 5$}
    \psfrag{-5}[][]{$\scriptstyle -5$}
    \psfrag{10}[][]{$\scriptstyle 10$}
    \psfrag{c8}[][]{$C_8$}
    \psfrag{c7}[][]{$C_7$}
    \psfrag{c6}[][]{$C_6$}
    \psfrag{c5}[][]{$C_5$}
    \psfrag{c4}[][]{$C_4$}
    \psfrag{c3}[][]{$C_3$}
    \psfrag{c2}[][]{$C_2$}
    \psfrag{c1}[][]{$C_1$}
    \includegraphics[width=0.6\textwidth]{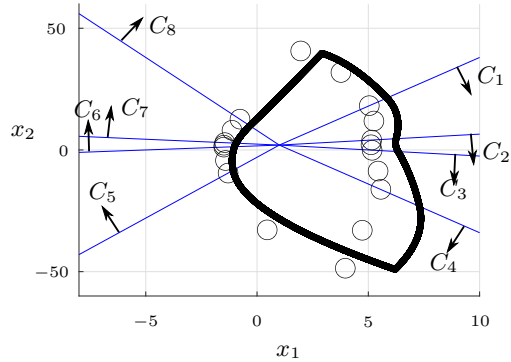}
    \caption{\cite{icinco25} Sample points (black circles), switching boundaries in solid blue, and the limit cycle (bold black line) of the switching system obtained for the model optimized with $n_d=4$.}
    \label{Example1}
\end{figure} To reduce these, the order is raised to $n_d=9$, and the solution of \eqref{OptP1}
with $T=1.4$ in the periodic signal \eqref{eq:signal} leads to the pairs $(A_i,b_i)$, $i\in\{1,\ldots,8\}$ as well as the limit cycle shown in black in Fig.~\ref{Example2}.

\begin{figure}[!t]
    \centering
    \small
    \psfrag{x1}[][]{$x_1$}
    \psfrag{x2}[][]{$x_2$}
    \psfrag{50}[][]{$\scriptstyle 50$}
    \psfrag{-50}[][]{$\scriptstyle -50$}
    \psfrag{0}[][]{$\scriptstyle 0$}
    \psfrag{5}[][]{$\scriptstyle 5$}
    \psfrag{-5}[][]{$\scriptstyle -5$}
    \psfrag{10}[][]{$\scriptstyle 10$}
    \psfrag{a8}[][]{$A_8,b_8$}
    \psfrag{a7}[][]{$A_7,b_7$}
    \psfrag{a6}[][]{$A_6,b_6$}
    \psfrag{a5}[][]{$A_5,b_5$}
    \psfrag{a4}[][]{$A_4,b_4$}
    \psfrag{a3}[][]{$A_3,b_3$}
    \psfrag{a2}[][]{$A_2,b_2$}
    \psfrag{a1}[][]{$A_1,b_1$}
    \includegraphics[width=0.6\textwidth]{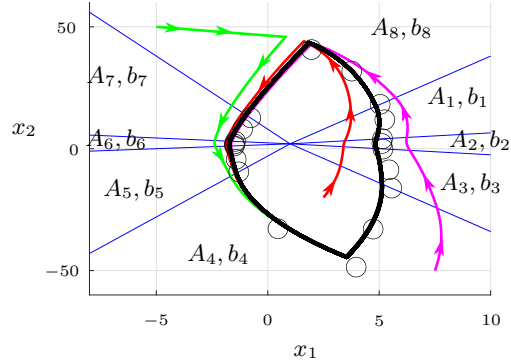}
    \caption{\cite{icinco25} Sample points (circles), switching boundaries (blue), limit cycle (bold black line), and trajectories from different initial points $x(0)=\begin{bmatrix}
        2.5&-20
    \end{bmatrix}^T$ red, $x(0)=\begin{bmatrix}
        -5&50
    \end{bmatrix}^T$ green, and $x(0)=\begin{bmatrix}
        7.5&-50
    \end{bmatrix}^T$ magenta, as obtained from the model synthesized with $n_d=9$.}
    \label{Example2}
\end{figure}
A significant reduction in the distance to the sampled points is achieved. The stability and uniqueness of the limit cycle are demonstrated by simulating trajectories from various initial points, initialized  inside and outside of the limit cycle. Note that the optional condition \eqref{OptP1c3} was omitted from this optimization, since the matrix exponential approximation with $n_d=9$ provided sufficient accuracy.
\begin{figure}[!ht]
    \centering
    \psfrag{x}[][]{$x$}
    \psfrag{time t}[][]{$t$}
    \includegraphics[width=0.6\textwidth]{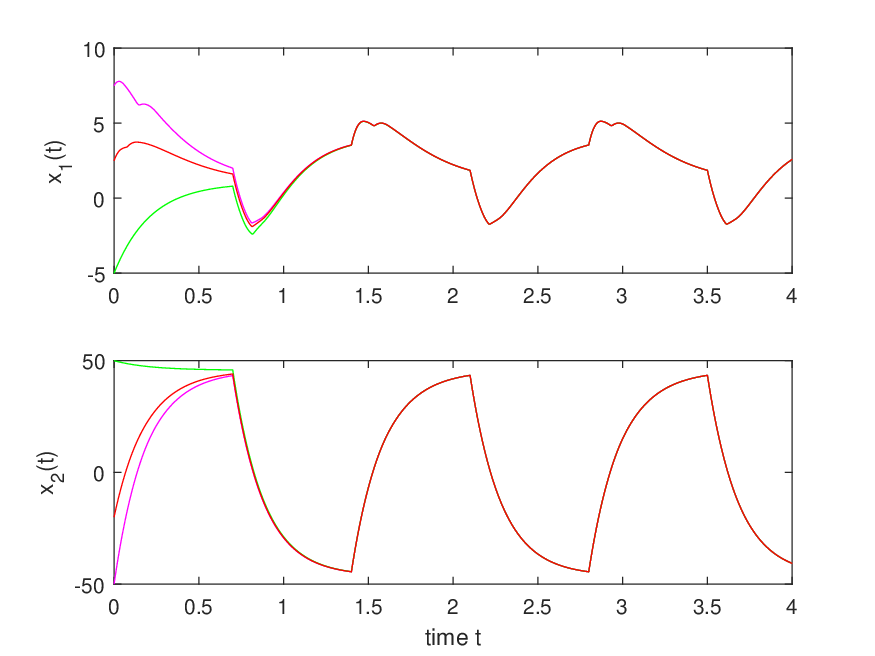}
    \caption{\cite{icinco25} Convergence of  $x_1(t)$ and $x_2(t)$ for different initial points $x(0)=\begin{bmatrix}
        2.5&-20
    \end{bmatrix}^T$ red, $x(0)=\begin{bmatrix}
        -5&50
    \end{bmatrix}^T$ green, and $x(0)=\begin{bmatrix}
        7.5&-50
    \end{bmatrix}^T$ magenta.}
    \label{Example2a}
\end{figure}

For a second test in a three-dimensional state space, the set of sample points shown in Fig.~\ref{Example3} (black circles) is examined. The procedure from Sec.~\ref{Partition} for $n_x=3$ was applied in conjunction with the first mentioned countermeasure (transition times between sample points on neighboring boundaries less than 1), leading to $n_p = 8$ subsets $F_i$ designed to minimize the transition time between each pair $\hat{x}_i$. The optimization problem \eqref{eq:optimalplane} was solved using the provided sample points and the center point $x_s = [1, 2, 0]^T$, resulting in:
\begin{align*}
    \Omega^*=\begin{bmatrix}
    -0.0115&-0.0066&0.9999
\end{bmatrix},~\epsilon^*=-0.0247,
\end{align*}
as shown by a green plane in Fig.~\ref{Example3}. The required line of intersection is $\Gamma:=\{x \in  \mathbb{R}^3 ~| ~ x= x_s + \eta \Omega^*, ~ \eta \in \mathbb{R} \}$, as marked by the solid blue line in Fig.~\ref{Example3}. Since no point in $F$ is contained in $\Gamma$, a unique plane
containing $\Gamma$ and $\hat x_i$ is determined for each sample point $\hat x_i$ to: 
\begin{align*}
&C_1=\begin{bmatrix}
    1&-\frac{1}{4}&0.0099
\end{bmatrix},~C_2=\begin{bmatrix}
    1&-2&-0.0017
\end{bmatrix},~C_3=\begin{bmatrix}
    -1&-2&-0.0247
\end{bmatrix},\\
&C_4=\begin{bmatrix}
    -1&-\frac{1}{4}&-0.0132
\end{bmatrix},~C_5=\begin{bmatrix}
    -1&\frac{1}{5}&-0.0102
\end{bmatrix},~C_6=\begin{bmatrix}
   -1&3&0.0083
\end{bmatrix}, \\
&C_7=\begin{bmatrix}
    1&\frac{5}{2}&0.0280
\end{bmatrix},~C_8=\begin{bmatrix}
    1&\frac{1}{6}&0.0126
\end{bmatrix},\\
&d_1\hspace{-1mm}=\hspace{-1mm}0.5,~d_2\hspace{-1mm}=\hspace{-1mm}-3,~d_3\hspace{-1mm}=\hspace{-1mm}-5,~d_4\hspace{-1mm}=\hspace{-1mm}-1.5,~d_5\hspace{-1mm}=\hspace{-1mm}-0.6,~d_6\hspace{-1mm}=\hspace{-1mm}5,~d_7\hspace{-1mm}=\hspace{-1mm}6,~d_8\hspace{-1mm}=\hspace{-1mm}1.33,
\end{align*}
see the solid blue planes in Fig.~\ref{Example3}.
\begin{figure}[!ht]
    \centering
    \psfrag{x}[][]{$x$}
    \psfrag{1}[][]{$1$}
    \psfrag{2}[][]{$2$}
    \psfrag{3}[][]{$3$}
    \includegraphics[width=0.6\textwidth]{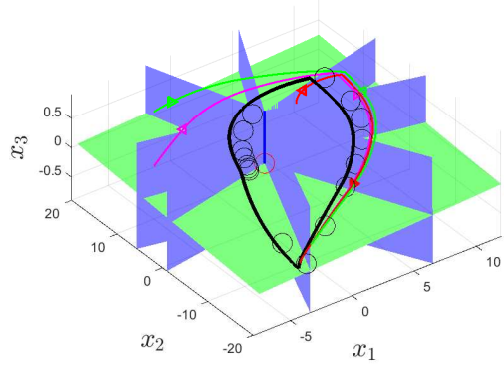}
    \caption{\cite{icinco25} Sample points (circles), green plane with minimum distance to the sampled points including $x_s$ (red circle), line of intersection $\Gamma$ (blue line), and boundaries shown as planes in blue, while the limit cycle is marked in black and the trajectories from different initial points $x(0)=\begin{bmatrix}
        6.5&10&0.1
    \end{bmatrix}^T$ in red, $x(0)=\begin{bmatrix}
        -1.9&18.5&0.18
    \end{bmatrix}^T$ green, and $x(0)=\begin{bmatrix}
        -5&10&0
    \end{bmatrix}^T$ magenta.}
    \label{Example3}
\end{figure}
The resulting limit cycle, obtained by solving the optimization problem from Section \ref{optimization}, is indicated in black in Fig.~\ref{Example3}. The distance to the sampling points remains acceptable, even though the dimensions operate on significantly different scales. This underscores the approximation quality of the approach. Trajectories from different initial points again demonstrate convergence to the unique limit cycle, as shown in Fig.~\ref{Example3} and Fig.~\ref{Example3a} (trajectories in red, green, and magenta).
\begin{figure}[!ht]
    \centering
    \psfrag{x}[][]{$x$}
    \psfrag{t}[][]{$t$}
    \includegraphics[width=0.6\textwidth]{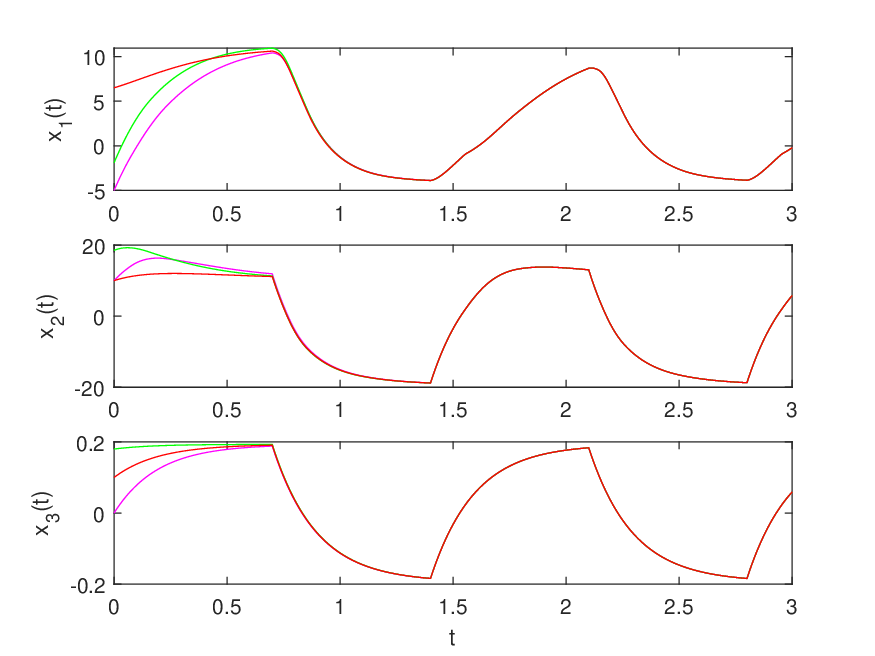}
    \caption{\cite{icinco25} Convergence of  $x_1(t),~\hspace{-0.5mm}x_2(t),~\hspace{-0.5mm}x_3(t)$ for different initial points $x(0)=\begin{bmatrix}
        6.5&10&0.1
    \end{bmatrix}^T$ red, $x(0)=\begin{bmatrix}
        -1.9&18.5&0.18
    \end{bmatrix}^T$ green, and $x(0)=\begin{bmatrix}
        -5&10&0
    \end{bmatrix}^T$ magenta.}
    \label{Example3a}
\end{figure}

\section{Reference Tracking of Periodic Switching Affine Systems}
If solving the optimization problem \eqref{OptP1}–\eqref{OptP1c4} for model synthesis leads to an unsatisfactory approximation result, the identification can be repeated without the stability constraints \eqref{OptP1c1} and \eqref{OptP1c2}, aiming only at the best fit to the set $F$. This simplification would be motivated by the fact that the previously proposed problem is inherently non-convex, highly nonlinear, and lacks a guaranteed solution -- difficulties that are partly caused by encoding the existence of a common Lyapunov function. While important for global stability of the identified limit cycle, the conditions \eqref{OptP1c1} and \eqref{OptP1c2} significantly deteriorate the approximation result. It is therefore reasonable to omit the conditions during the phase of model identification, and focus on achieving high accuracy of approximating the data. While the resulting identified periodic behavior $\bar{x}^{*}_{r,\lbrack0,\infty\lbrack}$ is no longer guaranteed to be globally stable, the underlying switched affine structure is amenable for a-posteriori design of stabilizing controllers. 
Thus, the method presented in this section leverages the affine structure to synthesizing  control laws which enables asymptotic tracking of the identified periodic trajectory by considering it as a reference. Thus, 
the reference tracking problem from Def.~\ref{eq:pwadeftracking}  is addressed for the system \eqref{eq:pwadeftracking} and a reference of type \eqref{eq:reference}.

Consider first the following convergence property from \cite{demidovich1967lectures} and \cite{pavlov2007convergence}:

\begin{lemma}  \label{lemma:contractive2} (Convergence  of Contractive Systems)
If the system \eqref{eq:pwadeftracking}  is contractive, then the solution $x_c(t)$, $t\geq 0$ for a fixed control signal $u_c(t)$ and  starting from an arbitrary initial state $x_c(0)\in \mathbb{R}^{n_x}$ always converges to a unique solution $\bar{x}_c(t)$ for $t \to +\infty$.  
\hfill$\Box$ 
\end{lemma}
Based on Lemma~\ref{lemma:contractive2}, the  control law:
\begin{align}  \label{eq:controllaw} 
u_c(t) := u_r(t) +K_ix_c(t)+w_i -(K_jx_r(t) +w_j)
\end{align} 
with $x_c(t) \in P_i$ and $x_r(t) \in P_j$, $i, j \in\{1, \ldots, n_p\}$ is used in \cite{van2008tracking} to track the reference trajectory $x_r(t)$. In here,  $K_i \in  \mathbb{R}^{n_u \times n_x}$ and $w_i \in  \mathbb{R}^{n_u \times 1}$ are parameters to be synthesized. Inserting \eqref{eq:controllaw} into  \eqref{eq:pwadeftracking} yields the controlled dynamics:
\begin{align}  \label{eq:controlledPWA} 
\dot{x}_c(t) \hspace{-1mm}=\hspace{-1mm}  (\hspace{-0.5mm}A_i\hspace{-1mm}+\hspace{-1mm}BK_i)x_c(t)\hspace{-1mm}+\hspace{-1mm}b_i\hspace{-1mm} +\hspace{-1mm}Bw_i \hspace{-1mm}+ \hspace{-1mm}B(\hspace{-0.5mm}u_r(t) \hspace{-0.8mm}-\hspace{-0.8mm} K_jx_r(t)\hspace{-0.8mm} -\hspace{-0.8mm}w_j\hspace{-0.5mm}),
\end{align} 
for which it can be noticed that the reference trajectory $x_r(t)$ governed by \eqref{eq:reference} is always a solution of  \eqref{eq:controlledPWA}. Thus, if the term $u_r(t) - (K_jx_r(t) +w_j)$ in   \eqref{eq:controlledPWA} is employed as the  control signal  $u_c(t)$ in  \eqref{eq:pwadeftracking} and if  \eqref{eq:controlledPWA} is contractive, it follows from Lemma~\ref{lemma:contractive2} that the solution from an arbitrary initial state $x_c(0)\in \mathbb{R}^{n_x}$  always converges to the reference $x_r(t)$, even if the latter is periodic. A straightforward approach to enforce this property
is to synthesize $K_i$ and $w_i$ by solving the tracking problem in Def.~\ref{eq:pwadeftracking}
while using the conditions 1 and 2 stated in Section~\ref{contractive} for the controlled dynamics:

1.) The derivative $\dot{x}_c(t)$ on any switching boundary of 
 \eqref{eq:controlledPWA} is  continuous, i.e.:
\begin{align}\label{eq:controlcontinuousgradient} 
(A_i\hspace{-0.5mm}  +\hspace{-0.5mm} BK_i)x_c\hspace{-0.5mm} +\hspace{-0.5mm} b_i \hspace{-0.5mm} +\hspace{-0.5mm}  Bw_i \hspace{-0.5mm}  \hspace{-0.5mm} =(A_{j}\hspace{-0.5mm} +\hspace{-0.5mm} BK_{j})x_c\hspace{-0.5mm} +\hspace{-0.5mm} b_{j} \hspace{-0.5mm} +\hspace{-0.5mm} Bw_{j}
\end{align}
 holds for all $x_c$ on any  boundary $C_{i}x_c  = d_{i}$, $i\in\{1, \ldots, n_p\}$ between adjacent polytopes $P_i$ and $P_{j}$. Note here that it is only required that the gradient of the controlled system  \eqref{eq:controlledPWA} is continuous, but not the one of the uncontrolled system in \eqref{eq:pwadeftracking}.
 
 2.) A positive-definite matrix $Q \in  \mathbb{R}^{{n_x} \times {n_x}} \succ 0$ exists such that:
\begin{align}  \label{eq:controlcommonlyapunov} 
(A_i+BK_i)^TQ+ Q(A_i+BK_i)\prec 0
\end{align}
holds for all  $i \in\{1, \ldots, n_p\}$. 

 It should be noted here that choosing the control parameters $K_i,K_j$ and $w_i,w_j$ n \eqref{eq:controlcontinuousgradient} and \eqref{eq:controlcommonlyapunov} to zero i would lead to the case discussed in Section \ref{contractive} for constructing the contractive dynamics. By Lemma~\ref{lemma:contractive} and if $u_r(t)$ is a piecewise continuous periodic signal, the dynamics \eqref{eq:controlledPWA} would thus itself obtain a globally stable limit cycle. Consequently, a system with a globally stable limit cycle obtained from the procedure in Section \ref{synthesis} is not suitable as reference system in the context of the addressed reference tracking problem. Evidently, no controller is required to track the periodic behavior when the reference itself already possesses a globally stable limit cycle. This section accordingly focuses on periodic references as switching affine systems, for which no stability statements can be made, with the goal of being able to track the reference asymptotically.

 The approach \cite{van2008tracking} in literature is recognized as conservative, since condition \eqref{eq:controlcommonlyapunov} requires a common matrix $Q$ across all polytopes. This is particularly true if the matrices $A_i$ differ substantially between the polytopes, or if the pairs $(A_i, B)$ are not all controllable. Thus, employing a control law of the form \eqref{eq:controlledPWA} can be severely restrictive, which motivates the development of a less demanding condition to ensure convergence.



\subsection{Synthesis with Guaranteed Convergence}

The synthesis scheme proposed below builds upon the approaches in \cite{johansson1997computation,miljkovic2025}, in which  piecewise quadratic Lyapunov functions were employed to prove (1) stability of PWA systems for a constant setpoint and reference tracking, as well as (2) a condition for state observation of a class of switching affine systems with specific space partitioning (parallel hyperplanes).
The following procedure extends these stability conditions to the tracking of periodic switching affine reference signals, based on the partitioning introduced in Section~\ref{Partition}. The tracking is achieved by combining the notion of contraction with the use of different Lyapunov functions assigned to the polytopes, thereby establishing convergence even when a common matrix $Q$ in \eqref{eq:controlcommonlyapunov} cannot be found.

Parameterized as described in Section~\ref{Partition}, the state space $\mathbb{R}^{n_x}$ is partitioned into a finite number of polytopes $P_i \subseteq \mathbb{R}^{n_x}$, with $i \in\{1, \ldots, n_P\}$; remember that this definition is particular in the sense that the number of polytopes $P_i$ and the number of bounding planes defined by $(C_i, d_i)$ are both equal to $n_p$. In addition the following is assumed for the periodic reference $x_r(t)$ and the state space partition:

\begin{assumption}\label{travelingtime} 
The  feed-forward control signal  $u_r(t)$ for the  reference trajectory $x_r(t)$ is selected such that $x_r(t)$ is contained in any polytope $P_i$, $i \in\{1, \ldots, n_P\}$ it passes through for a minimum dwell time of at least $T_{min} >0$.
\end{assumption}

\begin{assumption}\label{celldifference} 
When tracking the  periodic reference trajectory with $x_r(t) \in P_i$ at time $t$,  the state of \eqref{eq:controlledPWA} satisfies $x_c(t) \in \{P_{i-1} \cup P_{i} \cup  P_{i+1}\}$ for all $i \in\{1, \ldots, n_P\}$ and for all $t\in\mathbb{R}^{\geq 0}$.
\end{assumption}

\begin{assumption}\label{angle} 
    The bounding faces of any pair of two adjacent polytopes are defined such that the union of the two polytopes again forms a convex polytope.
\end{assumption}

The first assumption is easily met when the periodic reference $x_r(t)$ is obtained from the optimization procedure described earlier for identifying a switching affine system from a set of state samples. This is true since through the partitioning approach in Section~\ref{Partition} and the identification via the optimization problem in Section~\ref{optimization} (while omitting \eqref{OptP1c1} and \eqref{OptP1c2}), the limit cycle passes through any corresponding partition $P_i$ if the data approximation is satisfactory. Note that adjustments of $u_r(t)$ or the state space partitioning allow to enforce this effect.
Assumption 2 requires that the polytopes containing $x_c(t)$ and $x_r(t)$ are adjacent, and a relaxation of this assumption is made in Remark \ref{RemarkCommon}. Assumption 3 holds true if, e.g., the angle between the bounding faces of two adjacent polytopes $P_i,P_j$ is less than $\pi$ for all adjacent polytopes. The angle of $\pi/2$ between the two boundaries that define a single polytope would also satisfy assumption \ref{angle}, however, this is more restrictive with respect to the partitioning. In both cases, it then still holds that for any two states $x_c \in P_i$ and $x_r \in P_j$ ($P_i$ and $P_j$ sharing a common boundary), the line segment connecting $x_c$ and $x_r$ always intersects only with the boundary between $P_i$ and $P_j$. As illustrated in Fig. \ref{ExampleIntersection} for $n_x=2$ and $n_x=3$, the union of two neighboring polytopes thus always forms a convex polytope, and the number of switching lines for the reference tracking problem has to satisfy $n_P>4$. Note that an angular definition becomes nontrivial for $n_x \geq 4$. However, if the partitioning ensures the existence of a line segment analogous to the one in Fig. \ref{ExampleIntersection}, the upcoming theorem remains valid. The number of switching boundaries does not pose a practical challenge. In the special case that periodic behavior is generated by a bimodal switching affine system as in \cite{HS23,HLS24}, additional switching boundaries can simply be inserted such that the bounding faces between adjacent polytopes are smaller than $\pi$. To the resulting additional polytopes the same dynamics as in the original polytope is assigned. This leaves the overall behavior unchanged while artificially refining the partition to satisfy the requirement in Assumption \ref{angle}. 

Based on these assumptions and inspired by the previous results in \cite{miljkovic2025}, stability conditions for the reference tracking problem are stated next.
\begin{figure}[!ht]
    \centering
    \begin{minipage}[b]{0.475\textwidth}
        \centering
        \psfrag{x}[][]{$x_c(t)$}
    \psfrag{xt}[][r]{$\breve x(t)$}
    \psfrag{xr}[][r]{$x_r(t)$}
    \psfrag{q1}[][]{$P_{i-1}$}
    \psfrag{q}[][]{$P_i$}
    \psfrag{q2}[][]{$P_{i+1}$}
    \psfrag{C}[][]{$C_{i},d_i$}
        \includegraphics[width=0.75\textwidth, height=0.25\textheight]{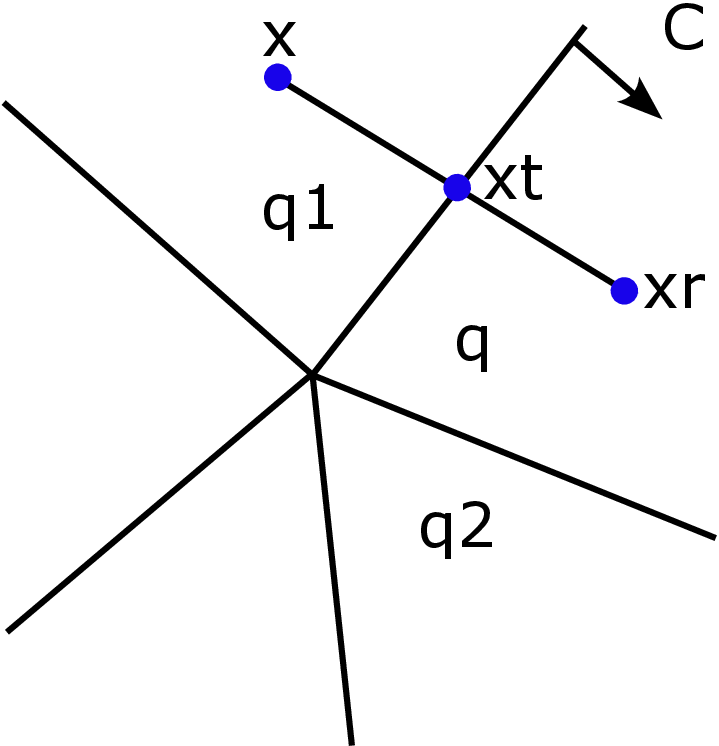}
        \subcaption{$n_x=2$}
        \label{ExampleIntersectiona}
    \end{minipage}
    \hfill
    \begin{minipage}[b]{0.475\textwidth}
        \centering
         \psfrag{x}[][r]{$x_c(t)$}
    \psfrag{xt}[][]{$\breve x(t)$}
    \psfrag{xr}[][r]{$x_r(t)$}
    \psfrag{q1}[][]{$P_{i-1}$}
    \psfrag{q}[][]{$P_i$}
    \psfrag{q2}[][]{$P_{i+1}$}
    \psfrag{c}[][r]{$C_{i},d_i$}
        \includegraphics[width=0.75\textwidth, height=0.25\textheight]{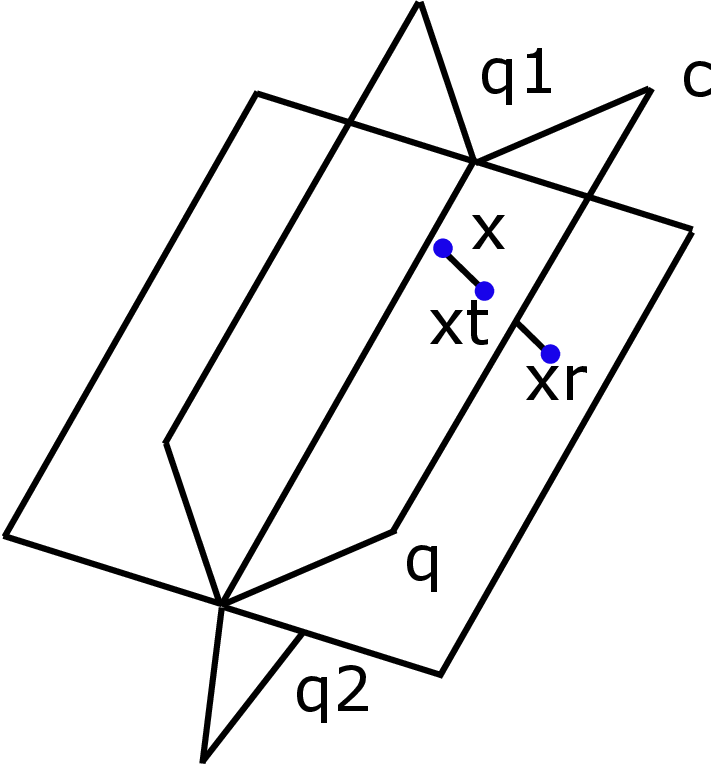}
        \subcaption{$n_x=3$}
        \label{ExampleIntersectionb}
    \end{minipage}
    \caption{The state $\breve x(t)$ denotes the intersection point on the boundary $C_{i}x=d_{i}$ of the line segment connecting $x_r(t)$ and $x_c(t)$.}
    \label{ExampleIntersection}
\end{figure}
\begin{theorem}\label{maintheorem} 
Given a system of type \eqref{eq:pwadeftracking}, a control law  \eqref{eq:controllaw} and a periodic reference trajectory $x_r(t)$ governed by \eqref{eq:reference}. Let the assumptions \ref{travelingtime}, \ref{celldifference} and \ref{angle}, and the continuity condition \eqref{eq:controlcontinuousgradient} on the switching boundaries hold. If then positive scalars $\rho>0\in \mathbb{R}$ and $\sigma>0\in \mathbb{R}$ as well as a set of positive-definite matrices $Q_i \in  \mathbb{R}^{{n_x} \times {n_x}} \succ 0$,  $i \in\{1, \ldots, n_P\}$ exist such that:
\begin{align} 
& (A_{i-1}+BK_{i-1})^TQ_i+ Q_i(A_{i-1}+BK_{i-1})\prec -\sigma Q_i \label{eq:condition1},  \\
& (A_{i}+BK_{i})^TQ_i+ Q_i(A_{i}+BK_{i})\prec-\sigma Q_i  \label{eq:condition2}, \\
& (A_{i+1}+BK_{i+1})^TQ_i+ Q_i(A_{i+1}+BK_{i+1})\prec -\sigma Q_i\label{eq:condition3}, \\
& Q_i \preceq \rho Q_{i+1}, ~ Q_{i+1} \preceq \rho Q_{i}, ~Q_i \preceq \rho Q_{i-1}, ~ Q_{i-1} \preceq \rho Q_{i}  \label{eq:condition4}, \\
&  \rho e^{-\sigma T_{min}} < 1  \label{eq:condition5} 
\end{align}
hold, then  $\lim_{t \to \infty}||x_c(t) - x_r(t)|| =0$ applies for any initial state $x_c(0)$.  \hfill$\Box$ 
\end{theorem}
\begin{proof}
First, a piecewise quadratic Lyapunov function $V_i(x_c(t), x_r(t)): \mathbb{R}^{2 n_x} \to \mathbb{R}$, with $i \in\{1, \ldots, n_P\}$, is introduced and defined as
\begin{align} \label{eq:lyapunov}
V_i(x_c(t), x_r(t)) &= (x_c(t) - x_r(t))^\top Q_i (x_c(t) - x_r(t)) \notag \\
&= \| x_c(t) - x_r(t) \|_{Q_i}^2, \quad \text{with}\ x_r(t) \in P_i,
\end{align}
which quantifies the distance between the state $x_c(t)$ and the reference trajectory $x_r(t)$ at time $t$. It is emphasized that the assignment of the matrix $Q_i$ in \eqref{eq:lyapunov} is determined solely by the polytope $P_i$ containing the reference $x_r(t)$, and not by the polytope containing the controlled state $x_c(t)$. Next the three different cases for adjacent polytopes containing $x_c(t)$ and $x_r(t)$ are analyzed with respect to the time derivatives of $V_i(x_c(t), x_r(t))$.

\underline{Case 1:} $x_r(t)\in P_i$ and $x_c(t)\in P_i$
\begin{align*} 
&\dot{V}_i(x_c(t), x_r(t)) = 2 (x_c(t) -x_r(t))^TQ_i (\dot{x}_c(t) -\dot{x}_r(t)) \\
&=2 (x_c(t) -x_r(t))^TQ_i (A_{i}+BK_{i}) (x_c(t) -x_r(t))   \\
&=  (x_c(t) - x_r(t))^T   (A_{i} + BK_{i})^T Q_i +  Q_i(A_{i} + BK_{i})(x_c(t)  - x_r(t)) \\ 
&< -\sigma  V_i(x_c(t), x_r(t))  < 0,
\end{align*}
where the last  two inequalities hold due to \eqref{eq:condition2}.

\underline{Case 2:} $x_r(t) \in P_i$, and $x_c(t)\in P_{i-1}$

In this setting, according to the partitioning of Section \ref{Partition} and Assumption \ref{angle}, a line segment connecting $x_c(t)$ and $x_r(t)$ necessarily intersects the boundary $C_{i}x = d_{i}$ between the polytopes $P_{i-1}$ and $P_i$ (see the illustration in Fig.~\ref{ExampleIntersection}). Denoting the intersection point by $\breve {x}(t)$, and since the three points $x_c(t)$, $\breve x(t)$ and $x_r(t)$ are located on the same line, the following equation hold:
\begin{align} \label{Intersectionpoint}
\frac{x_c(t)  -  \breve{x}(t)}{||x_c(t)  - \breve{x}(t)||_{Q_i}} = \frac{ \breve{x}(t)  -  x_r(t)  }{||\breve {x}(t)  -  x_r(t)||_{Q_i}}  =  \frac{x_c(t)  -  x_r(t)}{|| x_c(t)  -  x_r(t)||_{Q_i}}.
\end{align}
Based on the fact of continuity of the state derivatives on every switching boundary together with \eqref{Intersectionpoint}, the derivative of the piecewise quadratic Lyapunov function satisfies:
\begin{align} \label{eq:lyapunovderivative2} 
&\dot{V}_i(x_c(t), x_r(t)) = 2 (x_c(t) -x_r(t))^TQ_i (\dot{x}_c(t) -\dot{x}_r(t)) \notag \\
&=  2 (x_c(t) -x_r(t))^TQ_i (\dot{x}_c(t) - \dot{\breve{x}}(t) +  \dot{\breve{x}}(t) - \dot{x}_r(t)) \notag \\
&= 2 || x_c(t) -x_r(t)||_{Q_i} \left( \frac{ (x_c(t) - \breve{x}(t))^T Q_i (\dot{x}_c(t) - \dot{\breve{x}}(t))  }{||x_c(t) -\breve{x}(t)||_{Q_i}}\right. \notag \\
& + \left. \frac{ (\breve{x}(t)-x_r(t)) ^T Q_i ( \dot{\breve{x}}(t) - \dot{x}_r(t))  }{||\breve{x}(t) -x_r(t)||_{Q_i}}  \right).
\end{align}
From \eqref{eq:condition1} it is known that: 
 \begin{align} \label{eq:lyapunovderivative3} 
&(x_c(t) \hspace{-0.5mm}-\hspace{-0.5mm} \breve{x}(t))^T Q_i (\dot{x}_c(t) \hspace{-0.5mm}-\hspace{-0.5mm} \dot{\breve{x}}(t))= (x_c(t) \hspace{-0.5mm}-\hspace{-0.5mm} \breve{x}(t))^T Q_i (A_{i-1}\hspace{-0.5mm}\hspace{-0.5mm}+\hspace{-0.5mm}\hspace{-0.5mm}BK_{i-1}) (x_c(t)\hspace{-0.5mm} - \hspace{-0.5mm}\breve{x}(t))  \notag \\
  & < -\frac{\sigma}{2} ||x_c(t) -\breve{x}(t)||^2_{Q_i}
\end{align}
and from \eqref{eq:condition2} that: 
 \begin{align} \label{eq:lyapunovderivative4} 
&  (\breve{x}(t)-x_r(t)) ^T Q_i (\dot{\breve{x}}(t) - \dot{x}_r(t))\hspace{-0.5mm}= \hspace{-0.5mm}(\breve{x}(t)-x_r(t))^TQ_i (A_{i}+BK_{i}) (\breve{x}(t)-x_r(t))  \notag \\
  & <  -\frac{\sigma }{2} ||\breve{x}(t)-x_r(t)||^2_{Q_i}
\end{align}
hold. Thus, the time derivative of the Lyapunov function in \eqref{eq:lyapunovderivative2} satisfies:
\begin{align} \label{eq:lyapunovderivative5} 
&\dot{V}_i(x_c(t), x_r(t))  <  -\sigma  || x_c(t) -x_r(t)||_{Q_i}  (||x_c(t) -\breve{x}(t)||_{Q_i} + ||\breve{x}(t)-x_r(t)||_{Q_i}) \notag \\
&<-\sigma  || x_c(t) -x_r(t)||^2_{Q_i} = -\sigma  V_i(x_c(t), x_r(t)) <0,
\end{align}
in which the second inequality holds since:
 \begin{align} \label{eq:lyapunovderivative6} 
||x_c(t) - x_r(t)||_{Q_i}  =  ||x_c(t)  -  \breve{x}(t)||_{Q_i}   +  ||\breve{x}(t)-x_r(t)||_{Q_i}
 \end{align}
 applies for any $x_c(t)$, $\breve{x}(t)$, and $x_r(t)$  on the same line.
 
 For the case with $x_r(t)\in P_i$ and $x_c(t)\in P_{i+1}$, the inequality \eqref{eq:lyapunovderivative5} can be shown following the same way as in case 2 but by use of \eqref{eq:condition3}; for brevity this derivation is omitted. 

Consequently, if $x_r(t)\in P_i$ applies and  Assumption 2 holds, the value of $V_i(x_c(t), x_r(t))$ is strictly decreasing over time and satisfies:
 \begin{align} \label{eq:rate1} 
 V_i(x_c(t), x_r(t)) < e^{-\sigma (t - t_{i, s})} V_i(x_c(t_{i, s}), x_r(t_{i, s})), 
\end{align}
where $t_{i, s}$ denotes the time that $x_r(t)$ first enters into the polytope $P_i$. If the reference $x_r(t)$ remains within the polytope $P_i$ for all $t \ge t_{i, s}$, the strict decrease of $V_i(x_c(t), x_r(t))$ implies that
$ \lim_{t \to +\infty} V_i(x_c(t), x_r(t)) = 0$, which hence ensures the convergence of the state $x_c(t)$ to the reference trajectory $x_r(t)$.

In the alternative scenario, where the reference $x_r(t)$ transitions from $P_i$ into the subsequent polytope $P_{i+1}$ at time $t_{i+1, s}$, it is known that:
 \begin{align} \label{eq:rate2} 
& V_{i+1}(x_c(t_{i+1, s}), x_r(t_{i+1, s}))  \le \rho V_{i}(x(t_{i+1, s}), x_r(t_{i+1, s})) \notag \\
 & <  \rho  e^{-\sigma (t_{i+1, s} - t_{i, s})} V_i(x_c(t_{i, s}), x_r(t_{i, s}))  \notag \\
  &  <   \hspace{-0.5mm} \rho   e^{-\sigma T_{min}} V_i(x_c(t_{i, s}), x_r(t_{i, s}))  \hspace{-0.5mm}  <   \hspace{-0.5mm}  V_i(x_c(t_{i, s}), x_r(t_{i, s}))
\end{align} holds, under exploitation of the four inequalities from \eqref{eq:condition4}, \eqref{eq:rate1}, Assumption 1, and \eqref{eq:condition5} respectively. For the scenario in which the reference $x_r(t)$ transitions from polytope $P_i$ into the preceding polytope $P_{i-1}$, a set of inequalities analogous to those in \eqref{eq:rate2} can be derived in a similar manner. As established by the inequalities in \eqref{eq:rate2}, a switch in the polytope containing the reference $x_r(t)$ may cause an increase of the value of the piecewise Lyapunov function at the switching instant. However, it is guaranteed that the new value remains strictly smaller than the value at the previous switching instant \eqref{eq:rate2}. Consequently, the state $x_c(t)$ must converge to the reference trajectory $x_r(t)$ as $t \to \infty$. \hfill$\Box$
\end{proof}

\begin{remark}\label{RemarkCommon}
Assumption~\ref{celldifference}, which imposes a bound on the maximal difference of the polytope index  between the locations of $x_c(t)$ and $x_r(t)$ can be relaxed from one to larger values $z \in \mathbb{R}^{>0}$. In this generalized setting, only the conditions \eqref{eq:condition1} and \eqref{eq:condition3} need to be replaced by:
\begin{align} 
& (A_{i-j}+BK_{i-j})^\top Q_i + Q_i(A_{i-j}+BK_{i-j}) \prec -\sigma P_i \label{eq:condition1new} \\
& (A_{i+j}+BK_{i+j})^\top Q_i + Q_i(A_{i+j}+BK_{i+j}) \prec -\sigma P_i \label{eq:condition3new}
\end{align}
for all $j \in\{1, \ldots, z\}$. The limiting case is given by $z := n_P$, for which inequalities \eqref{eq:condition1new} and \eqref{eq:condition3new} imply the existence of a common Lyapunov function. \hfill$\Box$
\end{remark}
\subsection{Numerical Example for Reference Tracking}
Assume the state space $\mathbb{R}^2$ is partitioned into six polytopes, in accordance to Assumption \ref{angle} by the following switching hyperplanes:
\begin{align*}
&C_1=\begin{bmatrix}
    1&-\frac{1}{4}
\end{bmatrix},~C_2=\begin{bmatrix}
    1&-2
\end{bmatrix},~C_3=\begin{bmatrix}
    -1&-2
\end{bmatrix},\\
&C_4=\begin{bmatrix}
    -1&\frac{1}{5}
\end{bmatrix},~C_5=\begin{bmatrix}
    -1&3
\end{bmatrix},~C_6=\begin{bmatrix}
   1&2.5
\end{bmatrix},\\
&d_1\hspace{-1mm}=\hspace{-1mm}0.5,~d_2\hspace{-1mm}=\hspace{-1mm}-3,~d_3\hspace{-1mm}=\hspace{-1mm}-5,~d_4\hspace{-1mm}=\hspace{-1mm}-0.6,~d_5\hspace{-1mm}=\hspace{-1mm}5,~d_6\hspace{-1mm}=\hspace{-1mm}6,
\end{align*}. Let the reference switching affine system of type \eqref{eq:reference} be given by:
\begin{align*}
&A_1 \hspace{-0.5mm}=\hspace{-0.5mm} \begin{bmatrix} -1.9855 & -1.1502 \\ 1.9296 & -5.8722 \end{bmatrix}, 
A_2 \hspace{-0.5mm}=\hspace{-0.5mm} \begin{bmatrix} -4.6404 & 4.1596 \\ 1.5597 & -5.1324 \end{bmatrix},  
A_3 \hspace{-0.5mm}=\hspace{-0.5mm} \begin{bmatrix} -7.0100 & -0.5795 \\ 1.5295 & -5.1926 \end{bmatrix}, \\
&A_4 \hspace{-0.5mm}=\hspace{-0.5mm} \begin{bmatrix} -13.5623 & -0.5406 \\ 0.2063 & -5.5750 \end{bmatrix},  
A_5 \hspace{-0.5mm}= \hspace{-0.5mm}\begin{bmatrix} -10.5991 & -9.4302 \\ -0.1190 & -4.5992 \end{bmatrix}\hspace{-0.5mm},  
A_6 \hspace{-0.5mm}=\hspace{-0.5mm} \begin{bmatrix} -7.1370 & -0.7749 \\ -0.3981 & -5.2970 \end{bmatrix}, \\
&b_1 = \begin{bmatrix} -2.6396 \\ -8.7036 \end{bmatrix}, 
b_2 = \begin{bmatrix} -10.6043 \\ -9.8134 \end{bmatrix},  
b_3 = \begin{bmatrix} 1.2435 \\ -9.6627 \end{bmatrix}, \\
&b_4 = \begin{bmatrix} 7.7180 \\ -7.5747 \end{bmatrix}, 
b_5 = \begin{bmatrix} 22.5340 \\ -9.2010 \end{bmatrix},  
b_6 = \begin{bmatrix} 1.7613 \\ -7.5263 \end{bmatrix}, \\
&B = \begin{bmatrix} 25 & 0 \\ 0 & 250 \end{bmatrix}, u_r(t)=\begin{bmatrix} square(\frac{2\pi}{1.4}t) \\ square(\frac{2\pi}{1.4}t) \end{bmatrix},
\end{align*} in which $square$ means the puls wave changing between the values $1$ and $-1$ with the frequency $\frac{2\pi}{1.4}$. The reference trajectory $x_r(t)$ evolves into a limit cycle $\bar{x}^{*}_{r,\lbrack0,\infty\lbrack}$, which runs through all six polytopes, see Fig.\ref{Example4}. Note that the system does not satisfy Condition 1 of Sec.~\ref{contractive}, thus the contractivity property is lost, consequently there are no global stability guarantees for the limit cycle, motivating the use of a tracking controller. With the proposed tracking scheme and with an observed minimum dwell time $T_{min}=0.0098$ obtained from the simulation of $\bar{x}^{*}_{r,\lbrack0,\infty\lbrack}:=x_r(t)$ the continuity condition \eqref{eq:controlcontinuousgradient} and the conditions \eqref{eq:condition1}-\eqref{eq:condition5} are satisfied by:
\begin{figure}[!t]
    \centering   \includegraphics[width=0.75\textwidth, height=0.25\textheight]{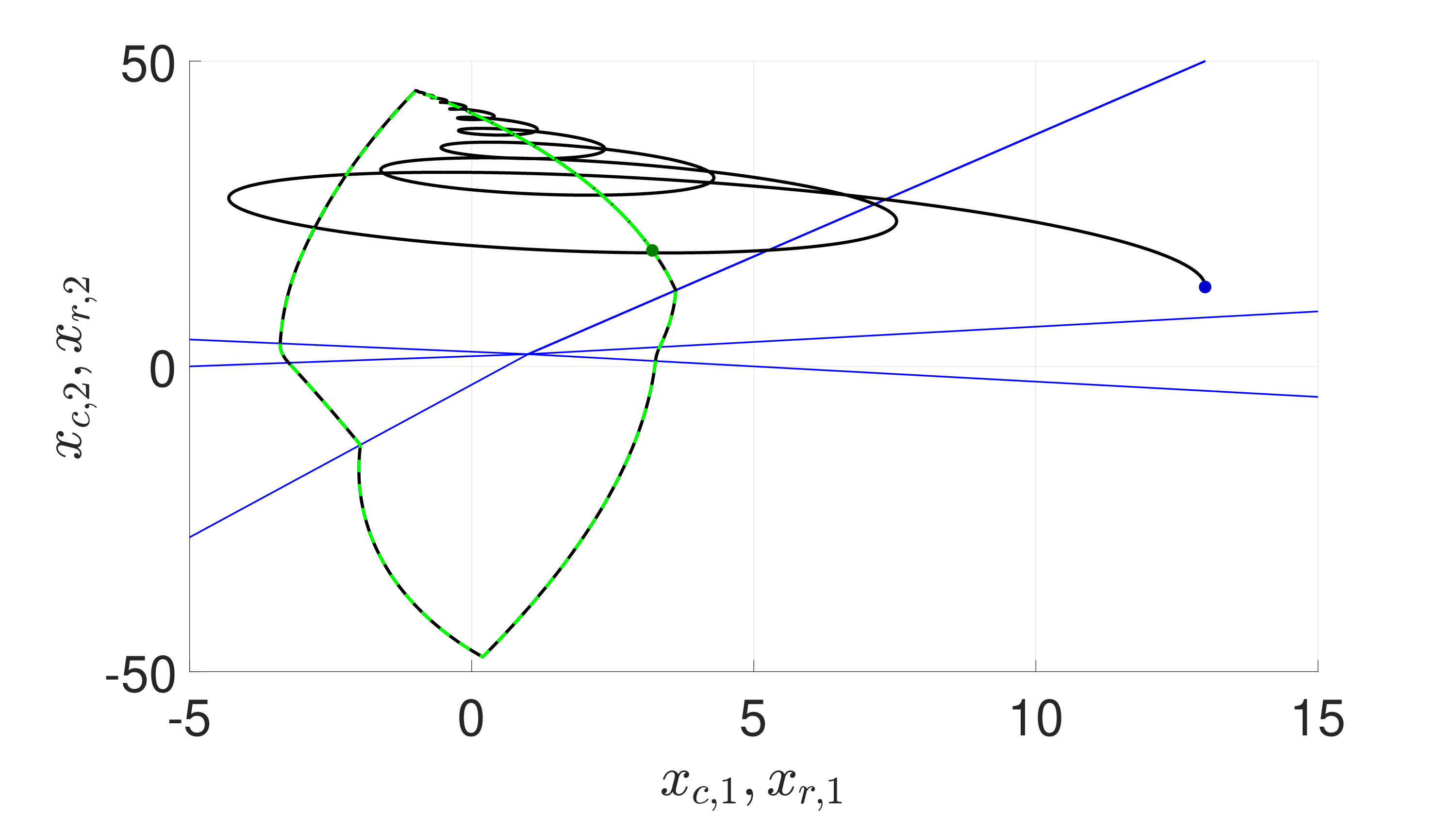}
    \caption{Evolution of the reference limit cycle (dashed green) and the state trajectory of the controlled system (black) for an initialization $x_r(0)=\begin{bmatrix} 3.2 && 19 \end{bmatrix}^T\in P_6$ and $x_c(0)=\begin{bmatrix} 13 && 13 \end{bmatrix}^T\in P_1$ in the partitioned state space (blue).}
    \label{Example4}
\end{figure}
\begin{align*}
&K_1 = \begin{bmatrix} -1.8541 & -3.2527 \\ 0.4939 & 0.1413 \end{bmatrix},  
K_2 = \begin{bmatrix} -1.7881 & -3.3849 \\ 0.4920 & 0.1451 \end{bmatrix},  
K_3 = \begin{bmatrix} -1.6949 & -3.1986 \\ 0.4903 & 0.1418 \end{bmatrix}, \\
&K_4 = \begin{bmatrix} -1.4879 & -3.1892 \\ 0.4862 & 0.1452 \end{bmatrix},  
K_5 = \begin{bmatrix} -1.5728 & -2.9344 \\ 0.4905 & 0.1323 \end{bmatrix}, 
K_6 = \begin{bmatrix} -1.7009 & -3.2545 \\ 0.4940 & 0.1412 \end{bmatrix}, \\
&w_1 = \begin{bmatrix} 0.2717 \\ -0.0029 \end{bmatrix},  
w_2 = \begin{bmatrix} 0.4699 \\ -0.0087 \end{bmatrix},  
w_3 = \begin{bmatrix} 0.0043 \\ -0.0005 \end{bmatrix}, \\
&w_4 = \begin{bmatrix} -0.2217 \\ -0.0032 \end{bmatrix},  
w_5 = \begin{bmatrix} -0.6463 \\ 0.0183 \end{bmatrix}, 
w_6 = \begin{bmatrix} 0.1221 \\ -0.0030 \end{bmatrix}, \\
&Q_1 = \begin{pmatrix} 17.5359 & 5.4649 \\ 5.4649 & 11.1105 \end{pmatrix}, 
Q_2 = \begin{pmatrix} 17.9853 & 5.5977 \\ 5.5977 & 11.1232 \end{pmatrix}, 
Q_3 = \begin{pmatrix} 17.5941 & 5.5231 \\ 5.5231 & 11.0175 \end{pmatrix}, \\
&Q_4 = \begin{pmatrix} 17.0113 & 5.3424 \\ 5.3424 & 10.9248 \end{pmatrix},  
Q_5 = \begin{pmatrix} 16.7391 & 5.2837 \\ 5.2837 & 10.9282 \end{pmatrix},  
Q_6 = \begin{pmatrix} 16.8744 & 5.3094 \\ 5.3094 & 11.0130 \end{pmatrix},
\end{align*} with $\rho=1.1154$, $\sigma=12.8307$ for the controlled system. As a result, asymptotic tracking behavior is guaranteed, as shown in Fig.~\ref{Example4} exemplarily for the initialization $x_r(0)=\begin{bmatrix} 3.2 && 19 \end{bmatrix}^T\in P_6$ and $x_c(0)=\begin{bmatrix} 13 && 13 \end{bmatrix}^T\in P_1$. The decreasing distance of the tracking error according to Def.~\ref{eq:pwadeftracking} is shown in Fig.~\ref{Example4a}. The evolution of $x_c(t)$ over time, illustrated in Fig.~\ref{Example4b} and Fig.~\ref{Example4c} for both components $x_{c,1}(t)$ and $x_{c,2}(t)$, shows that the determined controller tracks the periodic reference despite large deviations of the initialization, although the distance between $x_c(t)$ and $x_r(t)$ increases temporarily in the transient phase.

\begin{figure}[!t]
    \centering
    \includegraphics[width=0.75\textwidth, height=0.25\textheight]{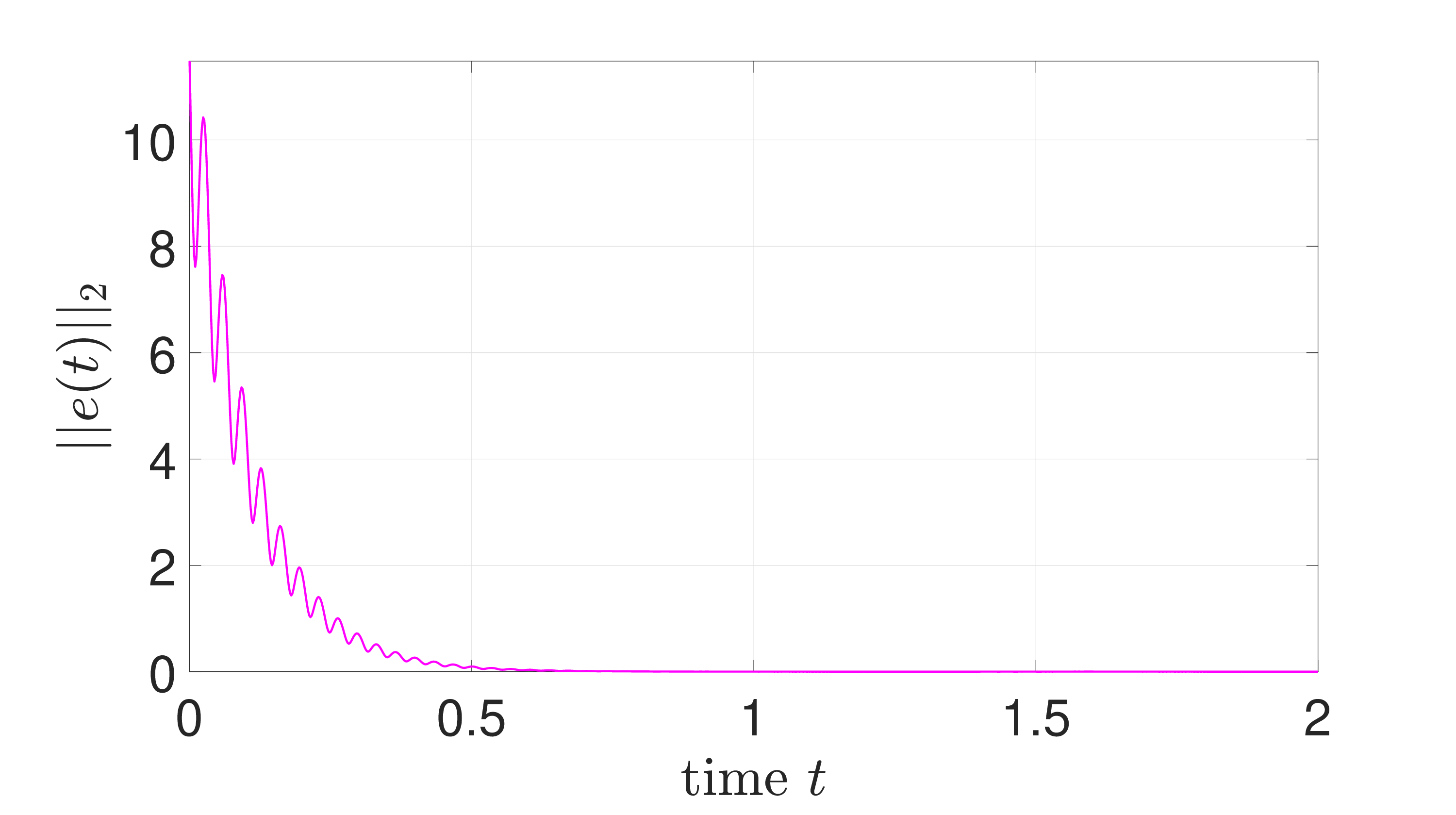}
    \caption{Evolution of the tracking error of the reference and the controlled system over time (magenta).}
    \label{Example4a}
\end{figure}
\begin{figure}[!t]
    \centering
    \includegraphics[width=0.75\textwidth, height=0.25\textheight]{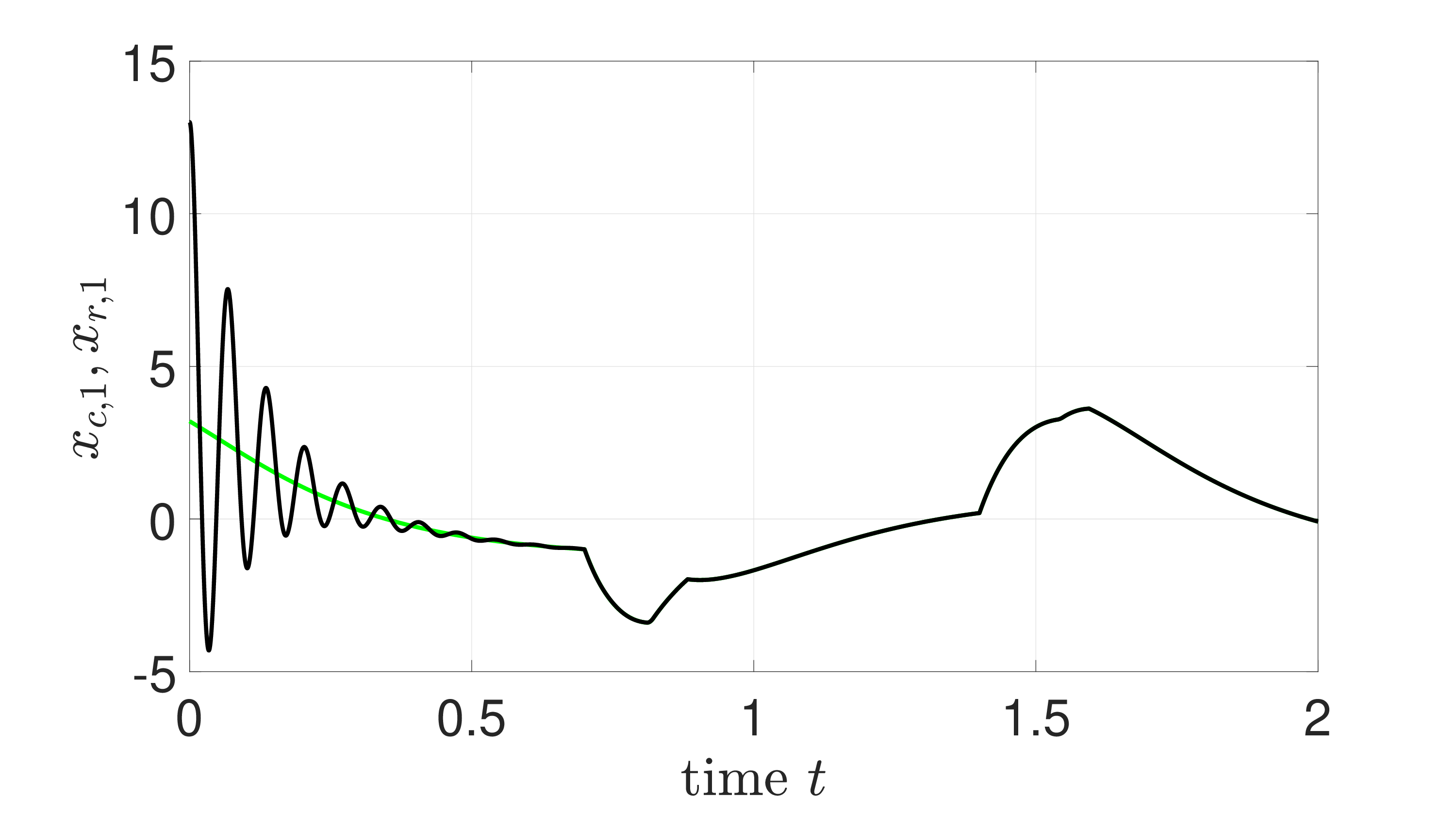}
    \caption{Evolution of $x_{r,1}(t)$ (green) and $x_{c,1}(t)$ (black).}
    \label{Example4b}
\end{figure}
\begin{figure}[!t]
    \centering
    \includegraphics[width=0.75\textwidth, height=0.25\textheight]{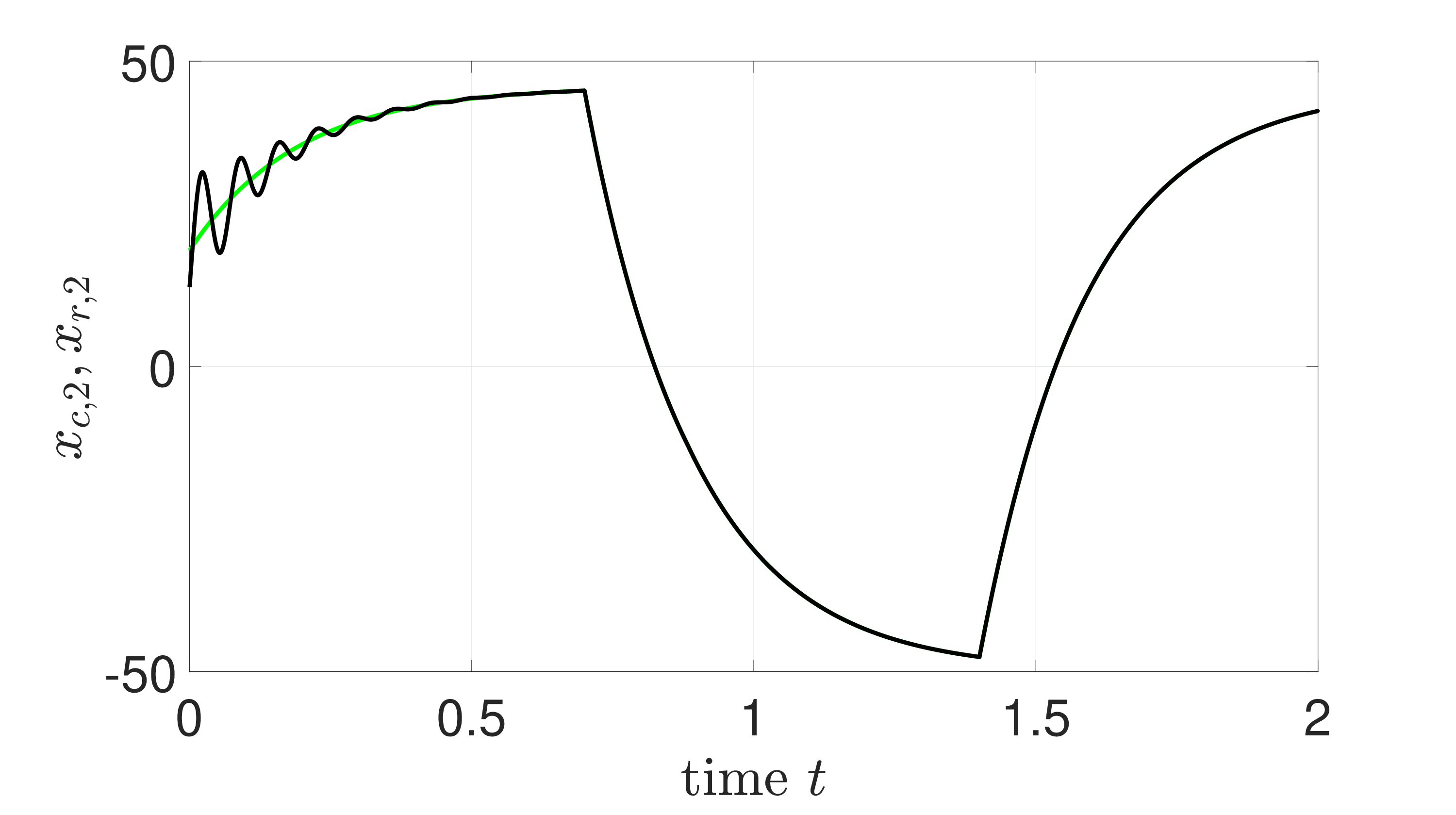}
    \caption{Evolution of $x_{r,2}(t)$ (green) and $x_{c,2}(t)$ (black).}
    \label{Example4c}
\end{figure}

\section{Conclusion}
This paper has introduced a method for approximating periodic behavior in nonlinear dynamical systems, extending beyond the planar case. By sampling the limit cycle of the nonlinear dynamics, switched affine systems with exogenous inputs are employed for approximation, preserving essential properties such as stability and uniqueness. In contrast to previously available methods, this approach provides constructive rules for partitioning the state space and synthesizing the dynamics through optimization, ensuring that the limit cycle of the switching affine system matches the sample points of the nonlinear data generator in an optimized sense. Through the notion of contractivity, global stability of the constructed limit cycle is guaranteed. Although the resulting conditions lead to a non-convex, nonlinear optimization problem, they simultaneously enforce smoothness of the limit cycle -- a characteristic typically observed in real-world oscillating systems. While obtaining the global solution to the optimization problem cannot be assured a-priori, any feasible solution yields a convergent trajectory. The use of a common intersection point for all boundaries proves particularly advantageous if sample points are evenly distributed around an interior region in $\mathbb{R}^{n_x}$.

Furthermore, conditions for asymptotic reference tracking in a class of switching affine systems with periodic solutions have been established, building on the property of contraction, and the use of multiple Lyapunov functions. Prior research has employed multiple Lyapunov functions for stationary setpoints. However, the present work is the first to apply them to reference tracking of periodic solutions in switching affine systems, using the specific partitioning addressed here. In cases where the identification procedure fails or yields an inadequate approximation of the observed dynamics (due to optimization constraints derived from the stability requirements), it is advisable to perform the identification first without such requirements, and then subsequently implement a reference tracking control scheme. As demonstrated in an example, these conditions enlarge the solution space for synthesizing a tracking controller, ensuring convergence to the reference.

Future work will investigate partitioning schemes without a common center point and the coupling of multiple oscillators of the proposed type. In addition, efficient synthesis methods for controller parameters satisfying the proposed conditions will be developed, aiming at formulating the design constraints in terms of a set of linear matrix inequalities.


\begin{credits}
\subsubsection{\ackname} Partial financial support by the German Research Foundation (DFG) through the Research Training Group \textit{Biological Clocks on Multiple Time Scales} (GRK 2749/1) is gratefully acknowledged.

\subsubsection{\discintname}
The authors have no competing interests to declare that are
relevant to the content of this article.
\end{credits}
%
%
%
\bibliographystyle{splncs04}
\bibliography{example}
\end{document}